%% file: main.tex
\documentclass[11pt]{article}

\usepackage{authblk}

\title{The I/O complexity of hybrid algorithms\\for square matrix multiplication}

\author{Lorenzo De Stefani\thanks{lorenzo@cs.brown.edu}}
\affil{Department of Computer Science, Brown University}

\usepackage[margin=1.25in]{geometry}

\usepackage{amsthm}
\newtheorem{theorem}{Theorem}
\newtheorem{lemma}[theorem]{Lemma}

\newtheorem{definition}{Definition}

\usepackage{pgf, tikz}
\usetikzlibrary{arrows, automata, petri}

\usepackage{longtable}
\usepackage{todonotes}

\usepackage{algorithm}
\usepackage[noend]{algpseudocode}
\usepackage[utf8]{inputenc}

\usepackage{subcaption}
\usepackage{amsmath}
\usepackage{amsfonts}
\usepackage{chato-notes}
\newcommand{\ri}{\mathcal{R}}
\newcommand{\BO}[1]{\mathcal{O}\left(#1\right)}
\newcommand{\BOme}[1]{\Omega\left(#1\right)}
\newcommand{\BT}[1]{\Theta\left( #1\right)}

\newcommand{\io}{I/O}

\newcommand{\alg}{\mathcal{A}}
\newcommand{\nproc}{P}
\newcommand{\dom}{D}

\newcommand{\setproof}{Q}

\newcommand{\globalInput}{\mathcal{X}}
\newcommand{\subpInput}{\mathcal{Y}}
\newcommand{\subpOutput}{\mathcal{Z}}

\newcommand{\msgsize}{B_m}

\newcommand{\nz}{n_0}

\newcommand{\hmm}{\mathfrak{H}}
\newcommand{\uhmm}{\mathfrak{UH}\left(\nz{}\right)}

\begin{document}
\begin{titlepage}
\pagenumbering{gobble}
\maketitle

\begin{abstract}
\input{abstract}
\end{abstract}
\end{titlepage}

\pagenumbering{arabic}

\section{Introduction}\label{sec:intro}
\input{introduction}

\section{Preliminaries}\label{sec:preliminaries}
\input{preliminaries}
\section{Hybrid matrix multiplication algorithms}\label{sec:algdef}
\input{algdef}
\section{\io{} lower bounds for algorithm in $\hmm{}$ and $\uhmm{}$}
\label{sec:mmlwb}
\input{hybmatmul}


\section{Conclusion}
\input{conclusion} 
\subsection*{Acknowledgments}
\input{ack}

\bibliographystyle{plain}
\bibliography{bibliography}
 \cleardoublepage
 \appendix
 \input{nappendix.tex}
\end{document}

%% file: abstract.tex
Asymptotically tight lower bounds are derived for the~\io{} complexity of a general class of hybrid algorithms computing the product of $n \times n$  square matrices combining  ``\emph{Strassen-like}''  fast matrix multiplication approach with  computational complexity $\BT{n^{\log_2 7}}$, and ``\emph{standard}'' matrix multiplication algorithms with computational complexity $\Omega\left(n^3\right)$. 
We present a novel and tight $\Omega\left(\left(\frac{n}{\max\{\sqrt{M},\nz{}\}}\right)^{\log_2 7}\left(\max\{1,\frac{\nz{}}{M}\}\right)^3M\right)$ lower bound for the \io{} complexity  a class of ``\emph{uniform, non-stationary}'' hybrid algorithms when executed in a two-level storage hierarchy with $M$ words of fast
memory, where $\nz{}$ denotes the threshold size of sub-problems  which are computed using standard algorithms with algebraic complexity  $\Omega\left(n^3\right)$.  

The lower bound is actually derived for the more
general class of ``\emph{non-uniform, non-stationary}'' hybrid algorithms which allow recursive calls to have a different structure, even when they refer to the multiplication of matrices of the same size and in the same recursive level, although the quantitative expressions become more involved. Our results are the first \io{} lower bounds for these classes of hybrid algorithms.
All presented lower bounds apply even if the
recomputation of partial results is allowed and are asymptotically tight.

The proof technique combines the analysis of the Grigoriev's
flow of the matrix multiplication function, combinatorial properties of the encoding functions used by fast Strassen-like algorithms, and an application of the Loomis-Whitney geometric theorem for the analysis of standard matrix multiplication algorithms.
Extensions of the lower bounds for a parallel model with $\nproc{}$ processors are also discussed.

%% file: introduction.tex
Data movement plays a critical role in  the performance of computing systems, in terms
of both time and energy. This technological
trend~\cite{patterson2005getting} appears destined to continue, as physical
limitations on minimum device size and on maximum message speed lead
to inherent costs when moving data, whether across the levels of a
hierarchical memory system or between processing elements of a
parallel system~\cite{bilardi1995horizons}. While the  communication
requirements of algorithms have been widely investigated in literature,
obtaining significant and tight lower bounds based on such requirements remains
an important and challenging task.

In this paper, we focus on the \io{} complexity of a general class of hybrid algorithms for the computing the product of square matrices which combine. Such algorithms combine  \emph{fast} algorithms with base case $2\times 2$ similar to  Strassen's matrix
multiplication algorithm~\cite{strassen1969gaussian} with algebraic (or \emph{computational}) complexity  $\BO{n^
{\log_2 7}}$ with \emph{standard} (or \emph{classic}) matrix multiplication algorithms with algebraic complexity $\BOme{n^
3}$.  Further, these algorithms allow recursive calls to have a different structure, even when
they refer to the multiplication of matrices in the same recursive level and of the same input size. These algorithms are referred in literature as ``\emph{non-uniform, non-stationary}''. This class includes, for example, algorithms that optimize for input sizes~\cite{desprez2004impact,douglas1994gemmw,huss1996implementation}. Matrix multiplication is a pervasive
primitive utilized in many applications.

While of actual practical importance, to the best of our knowledge, no characterization of the \io{} complexity of such algorithms has presented before this work. This is likely due to the the fact that the irregular nature of hybrid algorithms and, hence, the irregular structure of the corresponding Computational Directed Acyclic Graphs (CDAGs), complicates the analysis of the combinatorial properties of the CDAG which is the foundation of many of \io{} lower bound technique presented in literature (e.g.,~\cite{ballard2012graph,jia1981complexity,savage1995extending}).

The technique used in this work overcomes such challenges and yields asymptotically tight \io{} lower bounds which hold even if recomputation of intermediate values is allowed.

\paragraph{Previous and Related Work:} Strassen~\cite{strassen1969gaussian} showed that two $n \times n$
matrices can be multiplied with $O(n^{\omega})$ operations, where $\omega = \log_2
7 \approx 2.8074$, hence with asymptotically fewer than the $n^3$
arithmetic operations required by the straightforward implementation
of the definition of matrix multiplication. This result has motivated
a number of efforts which have lead to increasingly faster algorithms,
at least asymptotically, with the current record being at $\omega <
2.3728639$~\cite{legall2014}.

\io{} complexity has been introduced in the seminal work by Hong and
Kung~\cite{jia1981complexity}; it is essentially the number of data
transfers between the two levels of a memory hierarchy with a fast
memory of $M$ words and a slow memory with an unbounded number of
words. Hong and Kung presented techniques to develop lower bounds to the \io{}
complexity of computations modeled by \emph{computational directed acyclic graphs} (CDAGs). The resulting lower bounds
apply to all the schedules of the given CDAG, including those with
recomputation, that is, where some vertices of the CDAG are evaluated
multiple times. Among other results, they established a $\BOme{n^3/
  \sqrt{M}}$ lower bound to the \io{} complexity of standard, definition-based
matrix multiplication algorithms, which matched a known 
upper bound~\cite{cannon1969cellular}. The techniques
of~\cite{jia1981complexity} have also been extended to obtain tight
communication bounds for the definition-based matrix multiplication in
some parallel settings~\cite{ballard2012brief,irony2004communication,scquizzato2013communication} and for the special case of ``\emph{sparse matrix multiplication}''~\cite{pagh2014input}. Ballard et al.  generalized the results on matrix multiplication of
Hong and Kung~\cite{jia1981complexity} in~\cite{ballard2011minimizing,ballard2010communication} by using the approach proposed
in~\cite{irony2004communication} based on the Loomis-Whitney geometric
theorem~\cite{loomis1949,zalg}. 

In an important contribution, Ballard et
al.~\cite{ballard2012graph}, obtained an
$\Omega((n/\sqrt{M})^{\log_2 7}M)$ \io{} lower bound for Strassen's
algorithm, using the ``\emph{edge expansion approach}''. The authors
extend their technique to a class of ``\emph{Strassen-like}'' fast
multiplication algorithms and to fast recursive multiplication
algorithms for rectangular matrices~\cite{ballard2012graphrec}. This
result was later generalized to increasingly broader classes of
``\emph{Strassen-like}'' algorithms by Scott
et. al~\cite{scott2015matrix} using the ``\emph{path routing}''
technique, and De Stefani~\cite{thesis} using a combination the concept of Grigoriev's flow of a function and the ``\emph{dichotomy width}'' technique~\cite{bilardi1999processor}.
While the previously mentioned results hold only under the restrictive assumption that no intermediate result may be more than once (i.e., the \emph{no-recomputation assumption}), in~\cite{bilardi2017complexity} Bilardi and De Stefani introduced the first asymptotically tight \io{} lower bound which holds if recomputation is allowed. Their technique was later extended to the analysis of Strassen-like algorithms with base case $2\times 2$~\cite{ipdps2019}, and to the analysis of Toom-Cook integer multiplication algorithms~\cite{BilardiS19}.

A parallel, ``\emph{communication avoiding}''
implementation of Strassen's algorithm whose performance matches the
known lower bound~\cite{ballard2012graph,scott2015matrix}, was
proposed by Ballard et al.~\cite{ballard2012communicationalg}. A communication efficient algorithm for the special case of sparse matrices based on Strassen's algorithm was presented in~\cite{jacob2015fast}.

In~\cite{scott2015complexity}, Scott derived a lower bound for the \io{} complexity of a class of uniform, non-stationary algorithms combining Strassen-like algorithm with recursive standard algorithms. This result holds only under the restrictive no-recomputation assumption. 

To the best of our knowledge, ours is the first work presenting asymptotically tight \io{} lower bounds for non-uniform, non-stationary hybrid algorithms for matrix multiplication that hold when recomputation is allowed.


\paragraph*{Our results:}
We present the first \io{} lower bound for a class $\hmm{}$ of non-uniform, non-stationary hybrid matrix multiplication algorithms when executed in a two-level storage hierarchy with $M$ words of fast
memory. Algorithms in $\hmm{}$ combine fast Strassen-like algorithms with base case $2\times 2$ with algebraic complexity $\BT{n^{\log_2 7}}$, and standard algorithms based on the definition with algebraic complexity $\BOme{n^{3}}$. These algorithms allow recursive calls to have a different structure, even when they refer to the multiplication of matrices in the same recursive level and of the same input size. 
The result in Theorem~\ref{thm:genmatmul} relates the \io{} complexity of algorithms in $\hmm{}$ to the number and the input size of an opportunely selected set of the sub-problems generated by the algorithms themselves.

We also present, in Theorem~\ref{thm:corgenmatmul}, a novel $\Omega\left(\left(\frac{n}{\max\{\sqrt{M},\nz{}\}}\right)^{\log_2 7}\left(\max\{1,\frac{\nz{}}{M}\}\right)^3M\right)$ lower bound for the \io{} complexity of algorithms in a subclass $\uhmm{}$ of $\hmm{}$ composed by \emph{uniform non-stationary} hybrid algorithms where $\nz{}$ denotes the threshold size of sub-problems  which are computed using standard algorithms with algebraic complexity  $\BOme{n^3}$.

The previous result by Scott~\cite{scott2015complexity} covers only a sub-class of $\uhmm{}$ composed by uniform, non-stationary algorithms combining Strassen-like algorithms with the recursive standard algorithm, and holds only assuming that no intermediate value is recomputed. Instead, all our bounds allow for recomputation of intermediate values and are asymptotically tight. As the matching upper bounds do not recompute any intermediate value, we conclude that using recomputation may reduce the \io{} complexity of the considered classes of hybrid algorithms by at most a constant factor.

Our proof technique is of independent interest since it exploits to a significant extent the ``\emph{divide and conquer}'' nature exhibited by many algorithms. Our approach combines elements from the ``\emph{G-flow}'' \io{} lower bound technique originally introduced by Bilardi and De Stefani, with an application of the Loomis-Whitney geometric theorem, which has been used by Irony et al. to study the \io{} complexity of standard matrix multiplication algorithms~\cite{irony2004communication},  to recover an information which relates to the concept of \emph{Minimum set} introduced in Hong and Kung's method. We follow the dominator set approach pioneered by Hong and Kung
in~\cite{jia1981complexity}. However, we focus the dominator analysis
only on a select set of target vertices, which, depending on the algorithm structure, correspond either to the outputs of
the sub-CDAGs  that correspond to sub-problems of a
suitable size (i.e., chosen as a function of the fast memory
capacity $M$) computed using a fast Strassen-like algorithm, or to the the vertices corresponding to the elementary products evaluated by a standard (definition) matrix multiplication algorithm.

We derive our main results for the hierarchical memory model (or external memory model). Our results generalize to parallel models with~\nproc{} processors. For these parallel models, we derive lower bounds for the ``\emph{bandwith cost}'', that is for the number of messages (and, hence, the number of memory) that must be either sent or received by at least one processor during the CDAG evaluation.

\paragraph*{Paper organization:} In Section~\ref{sec:preliminaries} we outline the notation and the computational models used in the rest of the presentation. In Section~\ref{sec:algdef} we rigorously define the class of hybrid matrix multiplication algorithms $\hmm{}$ being considered. In Section~\ref{sec:CDAG} we discuss the construction and important properties of the CDAGs corresponding to algorithms in $\hmm{}$. In Section~\ref{sec:msp} we introduce the concept of Maximal Sup-Problem (MSP) and describe their properties which lead to the \io{} lower bounds for algorithms in $\hmm{}$ in Section~\ref{sec:mmlwb}.


%% file: preliminaries.tex
We consider algorithms that compute the product of two square matrices $\mathbf{A}\times \mathbf{B}=\mathbf{C}$ with entries from a ring $\ri{}$. We use $A$ to denote  the set variables each corresponding to an entry of matrix $\mathbf{A}$. We refer to the number of entries of a matrix $\mathbf{A}$ as its  ``\emph{size}'' and we denote it as $|A|$. We denote the entry on the $i$-th row of the $j$-th column of matrix $\mathbf{A}$ as $\mathbf{A}[i][j]$.   

In this work we focus on algorithms whose execution can be modeled
as a \emph{computational directed acyclic graph} (CDAG) $G=(V,E)$.
Each vertex $v\in V$ represents either an input value or the
result of a unit-time operation (i.e., an intermediate result or one
of the output values) which is stored using a single memory word. For example, each of the input (resp., output) vertices of $G$ corresponds to one of the $2n^2$ entries of the factor matrices $\mathbf{A}$ and $\mathbf{B}$ (resp., to the $n^2$ entries of the product matrix $\mathbf{C}$). The \emph{directed} edges in $E$ represent data dependencies. That is, a pair of vertices $u,v\in V$ are connected by an edge $(u,v)$ directed from $u$ to $v$ if and only if the value corresponding to $u$ is an operand of the unit time operation which computes the value corresponding to $v$.  A \emph{directed path} connecting vertices $u,v\in V$ is
an ordered sequence of vertices starting with $u$ and ending with $v$, such that there is an edge in $E$ directed from each vertex in the sequence to its successor. 

We say that $G'=(V',E')$ is a \emph{sub-CDAG} of $G=(V,E)$ if
$V'\subseteq V$ and $E' \subseteq E \cap (V'\times V')$. Note that, according to this definition, every CDAG is a sub-CDAG of itself. We say that two sub-CDAGs  $G'=(V',E')$ and $G''=(V'',E'')$ of $G$ are \emph{vertex disjoint} if $V'\cap V''=\emptyset$. Analogously, two directed paths in $G$ are vertex disjoint if they do not share any vertex.

When analyzing the properties of CDAGs we make use of the concept of \emph{dominator set} originally introduced in~\cite{jia1981complexity}. We use the following -- slightly different -- definition:

\begin{definition}[Dominator set]\label{def:dominator}
Given a CDAG $G=(V,E)$, let $I\subset V$ denote the set of its input
vertices. A set $\dom \subseteq V$ is a \emph{dominator set} for $V'\subseteq V$ with respect to $I'\subseteq I$ if every path from a vertex in $I'$ to a vertex in
$V'$ contains at least a vertex of $\dom$. 
When $I'=I$, $\dom$ is simply referred as ``a dominator set for $V'$''.
\end{definition}

\paragraph{\io{} model and machine models:}
We assume that sequential computations are executed on a system with a two-level memory hierarchy, consisting of a fast memory or \emph{cache} of size $M$ (measured in memory words) and a \emph{slow memory} of unlimited size. An operation can be executed only if all its operands are in cache. We assume that each entry of the input and intermediate results matrices (including entries of the output matrix) is maintained in a single memory word (the results trivially generalize if multiple memory words are used).



Data can be moved from the slow memory
to the cache by \texttt{read} operations and, in the other direction,
by \texttt{write} operations. These operations are also called
\emph{\io{} operations}. We assume the input data to be stored in slow
memory at the beginning of the computation. The evaluation of a CDAG in this model can be analyzed by means of the ``\emph{red-blue pebble game}''~\cite{jia1981complexity}. The number of \io{} operations executed when evaluating a CDAG depends on the ``\emph{computational schedule}'', that is, it depends on the order in which vertices are evaluated and on which values are kept in/discarded from cache. 

The \emph{\io{} complexity} $IO_{G}(M)$ of a CDAG $G$ is defined as the minimum number of \io{} operations over all possible computational schedules.
We further consider a generalization of this model known as the ``\emph{External Memory Model}'' by Aggarwal and Vitter~\cite{Aggarwal:1988:ICS:48529.48535},
where $B\geq 1$ values can be moved between cache and consecutive
slow memory locations with a single \io{} operation. For $B=1$, this model clearly reduces to the red-blue pebble game.

Given an algorithm $\mathcal{A}$, we only consider ``\emph{parsimonious execution schedules}'', that is schedules such that: (i) each time an intermediate result (excluding the output entries of $\mathbf{C}$) is computed, such value is then used to computed to compute at least one of the values of which it is an operand before being removed from the memory (either the cache or slow memory); and (ii) any time an intermediate result is read from slow to cache memory, such value is then used to computed to compute at least one of the values of which it is an operand before being removed from the memory or moved back to slow memory using a \emph{write} \io{} operation. Clearly, any non-parsimonious schedule $\mathcal{C}$ can be reduced to a parsimonious schedule $\mathcal{C}'$ by removing all the steps which violate the definition of parsimonious computation. $\mathcal{C}'$ has therefore less computational and \io{} operations than $\mathcal{C}$. Hence, restricting the analysis to parsimonious computations leads to no loss of generality.

We also consider a parallel model where $\nproc$ processors, each with
a local memory of size $2n^2/P\leq M<n^2$, are connected by a network. We do not, however, make any assumption on the initial distribution of the input data nor regarding the balance of the computational load among the $\nproc{}$ processors. Processors can exchange point-to-point
messages, with every message containing up to $\msgsize{}$ memory words. For this parallel model, we derive lower bounds for
the number of messages that must be either sent or received by at least
one processor during the CDAG evaluation. The notion of ``\emph{parsimonious execution schedules}'' straightforwardly extends to this parallel model.

%% file: algdef.tex
In this work, we consider a family of hybrid matrix multiplication algorithms obtained by hybridizing the two following classes of algorithms:\\

\noindent\textbf{Standard matrix multiplication algorithms:} This class includes all the square matrix multiplication algorithms which, given the input factor matrices $\mathbf{A},\mathbf{B} \in \ri{}^{n\times n}$, satisfy the following properties: 
    \begin{itemize}
        \item The $n^3$ \emph{elementary products}  $\mathbf{A}[i][j]\mathbf{B}[j][i]$, for $i,j=0,\ldots, n-1$, are \emph{directly computed};
        \item Each of the $\mathbf{C}[i][j]$ is computed by summing the values of the $n$ elementary products $\mathbf{A}[i][z]\mathbf{B}[z][j]$, for $z=0,\ldots, n-1$, through a \emph{summation tree} by additions and subtractions only;
        \item The evaluations of the $\mathbf{C}[i][j]$'s are \emph{independent of each other}. That is, internal vertex sets of the summation trees of all the $\mathbf{C}[i][j]$'s are \emph{disjoint from each other}. 
    \end{itemize}   
    This class, also referred in literature as \emph{classic}, \emph{naive} or \emph{conventional} algorithms, correspond to that studied for the by Hong and Kung~\cite{jia1981complexity} (for the sequential setting) and by Irony et al.~\cite{irony2004communication} (for the parallel setting). 
    Algorithms in this class have computational complexity $\Omega\left({n^3}\right)$. This class includes, among others, the sequential iterative \emph{definition} algorithm, the sequential recursive divide and conquer algorithm based on block partitioning, and parallel algorithms such as Cannon's ``\emph{2D}'' algorithm~\cite{cannon1969cellular}, the ``\emph{2.5D}'' algorithm by Solomonik and Demmel~\cite{solomonik2011communication}, and ``\emph{3D}'' algorithms~\cite{aggarwal1987model,johnsson1993minimizing}.\\

   \noindent\textbf{Fast Strassen-like matrix multiplication algorithms with base case $2\times 2$:} This class includes algorithms following a structure similar to that of Strassen's~\cite{strassen1969gaussian} (see Appendix~\ref{app:strassenalg}) and Winograd's variation~\cite{winograd1971multiplication} (which reduces the leading  coefficient of the arithmetic complexity reduced from
7 to 6). Algorithms in this class generally follow three steps: 
    \begin{enumerate}
        \item{\textbf{Encoding:}} Generate the inputs, of size $n/2\times n/2$  of seven \emph{sub-problems}, as linear sums of the input matrices;
        \item{\textbf{Recursive multiplications:}} Compute (recursively) the seven generated matrix multiplication sub-products;
        \item{\textbf{Decoding:}} Computing the entries of the product matrix $\mathbf{C}$ via linear combinations of the output of the seven sub-problems. 
    \end{enumerate}
    Algorithms in this class have algebraic complexity $\BO{n^{\log_2 7}}$, which is optimal for algorithms with base case $2\times 2$~\cite{winograd1971multiplication}.\\
    
Remarkably, the only properties of relevance for the analysis of the \io{} complexity of algorithms in these classes are those used in the characterization of the classes themselves.\\

In this work we consider a general class of \emph{non-uniform, non-stationary} hybrid square matrix multiplication algorithms, which allow  mixing of schemes from the fast Strassen-like class with algorithms from the standard class. Given an algorithm $\mathcal{A}$ let $P$ denote the problem corresponding to the computation of the product of the input matrices $\mathcal{A}$ and $\mathcal{B}$. Consider an ``\emph{instruction function}''  $f_\mathcal{A}(P)$, which, given as input $P$ returns either (a) indication regarding the algorithm from the standard class which is to be used to compute $P$, or (b) indication regarding the fast Strassen-like algorithm to be used to recursively generate seven sub-problems $P_1,P_2,\ldots,P_7$ and the instruction functions $f_\mathcal{A}(P_i)$ for each of the seven sub-problems. We refer to the class of non-uniform, non-stationary algorithms which can be characterized by means of such instruction functions as $\hmm{}$. Algorithms in $\hmm{}$ allow recursive calls to have a different structure, even when
they refer to the multiplication of matrices in the same recursive level. E.g., some of the sub-problems with the same size may be  computed using algorithms form the standard class while others may be computed using recursive algorithms from the fast class. This class includes, for example, algorithms that optimize for input sizes, (for sizes that are not an integer power of a constant integer). 


We also consider a sub-class $\uhmm {}$ of $\hmm{}$ constituted by \emph{uniform, non-stationary} hybrid algorithms which allow  mixing of schemes from the fast Strassen-like class for the initial $\ell$ recursion levels, and then cut the recursion off once the size the generated sub-problems is smaller or equal to a set threshold $\nz{}\times\nz{}$, and switch to using algorithm form the standard class. Algorithms in this class are \emph{uniform}, i.e., sub-problems of the same size are all either recursively computed using a scheme form the fast class, or are all computed using algorithms from the standard class.

This corresponds to actual practical scenarios, as the use of Strassen-like algorithms is  mostly beneficial for large input size. As the size of the input of the recursively generated sub-problems decreases, the asymptotic benefit of fast algorithms is lost due to the increasing relative impact of the constant multiplicative factor, and algorithms in the standard class exhibit lower \emph{actual} algebraic complexity. For a discussion on such hybrid algorithms and their implementation issues we refer the reader to~\cite{douglas1994gemmw,huss1996implementation}
(sequential model) and~\cite{desprez2004impact} (parallel model).
\section{The CDAG of algorithms in $\hmm{}$}\label{sec:CDAG}
Let $G^{\mathcal{A}}= (V_\mathcal{A},E_\mathcal{A})$ denote the CDAG that corresponds to an algorithm $\mathcal{A}\in\hmm{}$ used to multiply  input  matrices $\mathbf{A}, \mathbf{B}\in \ri{}^{n \times n}$. The challenge in the characterization of $G^{\mathcal{A}}$ comes from the fact that rather than considering to a single algorithm, we want to characterize the CDAG corresponding to the class $\hmm{}$. Further, the class $\hmm{}$ is composed by a rich variety of vastly different and irregular algorithm. Despite such variety, we show a general template for the construction of  $G^{\mathcal{A}}$ and we identify some of it properties which crucially hold \emph{regardless} of the implementation details of $\mathcal{A}$ and, hence, of  $G^{\mathcal{A}}$.
\paragraph*{Construction:}
$G^{\mathcal{A}}$ can be obtained by using a recursive construction that mirrors the recursive structure of the algorithm itself. Let $P$ denote the entire matrix multiplication problem computed by $\mathcal{A}$. 
Consider the case for which, according to the \emph{instruction function} $f_{\mathcal{A}}(P)$, $P$ is to be computed using an algorithm from the standard class. As we do not fix a specific algorithm, we do not correspondingly have a fixed CDAG. The only feature of interest for the analysis is that, in this case, the CDAG $G^{\mathcal{A}}$  corresponds to the execution of an algorithm from the standard class for input matrices of size $n \times n$.

Consider instead the case for which, according to $f_{\mathcal{A}}(P)$, $P$ is to be computed using an algorithm from the fast class. In the base case for $n=2$ the problem $P$ is computed without generating any further sub-problems. As an example, we present in Figure~\ref{fig:bsestra} the base case for Strassen's original algorithm~\cite{strassen1969gaussian}. If $n>2$, then $f_{\mathcal{A}}(P)$ specifies the divide and conquer scheme to be used to generate the seven sub-problems $P_1,P_2,\ldots,P_7$, and the \emph{instruction function} for each of them. The sub-CDAGs of $G^{\mathcal{A}}$ corresponding to each of the seven sub-problems $P_i$, denoted as $G^{\mathcal{A}}_{P_i}$ are constructed according to $f_{\mathcal{A}}(P_i)$, following recursively the steps discussed previously.
$G^{\mathcal{A}}$ can then be constructed by composing the seven sub-CDAGs $G^{\mathcal{A}}_{P_i}$.  
  $n^2$ disjoint copies of an \emph{encoder sub-CDAG} $Enc_A$ (resp., $Enc_B$) are used to connect the input vertices of $G^{2n \times 2n}$, which correspond to the values of the input matrix $\mathbf{A}$ (resp., $\mathbf{B}$) to the appropriate input vertices of the seven sub-CDAGs $G^{\mathcal{A}}_{P_i}$; the output vertices of the sub-CDAGs $G^{\mathcal{A}}_{P_i}$ (which correspond to the outputs of the seven sub-products) are connected to the appropriate output vertices of the entire $G^{\mathcal{A}}$ CDAG  using $n^2$ copies of the decoder sub-CDAG $Dec$.  We present an example of such recursive construction in Figure~\ref{fig:strarec}.

  \begin{figure}[bt]
  \begin{subfigure}{.48\textwidth}
  \centering
     \resizebox{\linewidth}{!}{
     \begin{tikzpicture}[scale = 0.5,
            > = stealth, 
            shorten > = 1pt, 
            auto,
            node distance = 3cm, 
            semithick 
        ]

        \tikzstyle{every state}=[
            draw = black,
            thick,
            minimum size = 4mm
        ]
		\node[state] at (-1,0) (A11) [label=below:$\mathbf{A}{[0]}{[0]}$] {};
        \node[state] at (1.7,0) (A12)  [label=below:$\mathbf{A}{[0]}{[1]}$] {};
        \node[state] at (4.3,0) (A21)  [label=below:$\mathbf{A}{[1]}{[0]}$] {};
        \node[state] at (7,0) (A22)  [label=below:$\mathbf{A}{[1]}{[1]}$] {};
        \node[state] at (-3,4) (A7)  [fill = blue] {};
        \node[state] at (-1,4) (A5)  [fill = blue] {};
        \node[state] at (1,4) (A4)    {};
        \node[state] at (3,4) (A1)   [fill = blue]{};
        \node[state] at (5,4) (A3)    {};
        \node[state] at (7,4) (A2)   [fill = blue] {};
        \node[state] at (9,4) (A6)   [fill = blue] {};

        \path[->] (A11) edge  (A5);
        \path[->] (A11) edge  (A1);
        \path[->] (A11) edge  (A3);
        \path[->] (A11) edge  (A6);
        
        \path[->] (A12) edge  (A7);
        \path[->] (A12) edge  (A5);
        
        \path[->] (A21) edge  (A2);
        \path[->] (A21) edge  (A6);
        
        \path[->] (A22) edge  (A2);
        \path[->] (A22) edge  (A1);
        \path[->] (A22) edge  (A4);
        \path[->] (A22) edge  (A7);
        
        \node[state] at (13,0) (B11) [label=below:$\mathbf{B}{[0]}{[0]}$] {};
        \node[state] at (15.7,0) (B12)  [label=below:$\mathbf{B}{[0]}{[1]}$] {};
        \node[state] at (18.3,0) (B21)  [label=below:$\mathbf{B}{[1]}{[0]}$] {};
        \node[state] at (21,0) (B22)  [label=below:$\mathbf{B}{[1]}{[1]}$] {};
        \node[state] at (11,4) (B7)  [fill = blue] {};
        \node[state] at (13,4) (B5)  {};
        \node[state] at (15,4) (B4)  [fill = blue] {};
        \node[state] at (17,4) (B1)  [fill = blue] {};
        \node[state] at (19,4) (B3)  [fill = blue] {};
        \node[state] at (21,4) (B2)   {};
        \node[state] at (23,4) (B6)  [fill = blue] {};

        \path[->] (B11) edge  (B4);
        \path[->] (B11) edge  (B1);
        \path[->] (B11) edge  (B2);
        \path[->] (B11) edge  (B6);
        
        \path[->] (B12) edge  (B3);
        \path[->] (B12) edge  (B6);
        
        \path[->] (B21) edge  (B7);
        \path[->] (B21) edge  (B4);
        
        \path[->] (B22) edge  (B7);
        \path[->] (B22) edge  (B1);
        \path[->] (B22) edge  (B5);
        \path[->] (B22) edge  (B3);

        \node[state] at (6,12) (C11) [label=above:$\mathbf{C}{[0]}{[0]}$] {};
        \node[state] at (8.7,12) (C12)  [label=above:$\mathbf{C}{[0]}{[1]}$] {};
        \node[state] at (11.3,12) (C21)  [label=above:$\mathbf{C}{[1]}{[0]}$] {};
        \node[state] at (14,12) (C22)  [label=above:$\mathbf{C}{[1]}{[1]}$] {};
        \node[state] at (2.5,8) (7)  [fill = red, label=left:$M_7$] {};
        \node[state] at (5,8) (5)  [fill = red, label=left:$M_5$] {};
        \node[state] at (7.5,8) (4)  [fill = red, label=left:$M_4$] {};
        \node[state] at (10,8) (1)  [fill = red, label=left:$M_1$] {};
        \node[state] at (12.5,8) (3)  [fill = red, label=left:$M_3$] {};
        \node[state] at (15,8) (2)  [fill = red, label=left:$M_2$] {};
        \node[state] at (17.5,8) (6)  [fill = red, label=left:$M_6$] {};
        
        \node[state, draw= white] at (7,1.5) (ena)  [label=right: \LARGE $Enc_A$] {};
        \node[state, draw= white] at (20,1.5) (enb)  [label=right: \LARGE $Enc_B$] {};
        \node[state, draw= white] at (0,10.5) (dec)  [label=right: \LARGE$Dec$] {};

        \path[<-] (C11) edge  (4);
        \path[<-] (C11) edge  (1);
        \path[<-] (C11) edge  (7);
        \path[<-] (C11) edge  (5);
        
        \path[<-] (C12) edge  (3);
        \path[<-] (C12) edge  (5);
        
        \path[<-] (C21) edge  (2);
        \path[<-] (C21) edge  (4);
        
        \path[<-] (C22) edge  (6);
        \path[<-] (C22) edge  (1);
        \path[<-] (C22) edge  (2);
        \path[<-] (C22) edge  (3);
        
        \path[->] (A7) edge  (7);
        \path[->] (B7) edge  (7);
        \path[->] (A5) edge  (5);
        \path[->] (B5) edge  (5);
        \path[->] (A4) edge  (4);
        \path[->] (B4) edge  (4);
        \path[->] (A1) edge  (1);
        \path[->] (B1) edge  (1);
        \path[->] (A3) edge  (3);
        \path[->] (B3) edge  (3);
        \path[->] (A2) edge  (2);
        \path[->] (B2) edge  (2);
        \path[->] (A6) edge  (6);
        \path[->] (B6) edge  (6);
    \end{tikzpicture}
     }
     \caption{$G^{\mathcal{A}}$ CDAG for base case $n=2$, using Strassen's algorithm~\cite{strassen1969gaussian} (see Appendix~\ref{app:strassenalg}).}
     \label{fig:bsestra}
  \end{subfigure}	
  \hspace{1mm}
  \begin{subfigure}{.49\textwidth}
  \centering
     \resizebox{\linewidth}{!}{
     	\begin{tikzpicture}[
            > = stealth, 
            shorten > = 1pt, 
            auto,
            node distance = 2cm, 
            semithick 
        ]

        \tikzstyle{every state}=[
            draw = black,
            thick,
            minimum size = 8mm
        ]
		\node[state, tokens=4] at (0,0) (A11) [label=below: \huge $A_{1,1}$] {};
        \node[state, tokens=4] at (2,0) (A12)  [label=below:\huge $A_{1,2}$] {};
        \node[state, tokens=4] at (4,0) (A21)  [label=below:\huge $A_{2,1}$] {};
        \node[state, tokens=4] at (6,0) (A22)  [label=below:\huge $A_{2,2}$] {};
        \node[state, draw = blue, tokens=4] at (-3,4) (A7)   {};
        \node[state, draw = blue, tokens=4] at (-1,4) (A5)   {};
        \node[state, tokens=4] at (1,4) (A4)    {};
        \node[state, draw = blue, tokens=4] at (3,4) (A1)    {};
        \node[state, tokens=4] at (5,4) (A3)    {};
        \node[state, draw = blue, tokens=4] at (7,4) (A2)    {};
        \node[state, draw = blue, tokens=4] at (9,4) (A6)    {};

        \path[->] (A11) edge  (A5);
        \path[->] (A11) edge  (A1);
        \path[->] (A11) edge  (A3);
        \path[->] (A11) edge  (A6);
        
        \path[->] (A12) edge  (A7);
        \path[->] (A12) edge  (A5);
        
        \path[->] (A21) edge  (A2);
        \path[->] (A21) edge  (A6);
        
        \path[->] (A22) edge  (A2);
        \path[->] (A22) edge  (A1);
        \path[->] (A22) edge  (A4);
        \path[->] (A22) edge  (A7);
        
        \node[state, tokens=4] at (14,0) (B11) [label=below:\huge $B_{1,1}$] {};
        \node[state, tokens=4] at (16,0) (B12)  [label=below:\huge $B_{1,2}$] {};
        \node[state, tokens=4] at (18,0) (B21)  [label=below:\huge $B_{2,1}$] {};
        \node[state, tokens=4] at (20,0) (B22)  [label=below:\huge $B_{2,2}$] {};
        \node[state, draw = blue, tokens=4] at (11,4) (B7)   {};
        \node[state, tokens=4] at (13,4) (B5)  {};
        \node[state, draw = blue, tokens=4] at (15,4) (B4)   {};
        \node[state, draw = blue, tokens=4] at (17,4) (B1)  {};
        \node[state, draw = blue, tokens=4] at (19,4) (B3)   {};
        \node[state, tokens=4] at (21,4) (B2)   {};
        \node[state, draw = blue, tokens=4] at (23,4) (B6)   {};

        \path[->] (B11) edge  (B4);
        \path[->] (B11) edge  (B1);
        \path[->] (B11) edge  (B2);
        \path[->] (B11) edge  (B6);
        
        \path[->] (B12) edge  (B3);
        \path[->] (B12) edge  (B6);
        
        \path[->] (B21) edge  (B7);
        \path[->] (B21) edge  (B4);
        
        \path[->] (B22) edge  (B7);
        \path[->] (B22) edge  (B1);
        \path[->] (B22) edge  (B5);
        \path[->] (B22) edge  (B3);

        \node[state, tokens=4] at (7,12) (C11) [label=above:\Huge $C_{1,1}$] {};
        \node[state, tokens=4] at (9,12) (C12)  [label=above:\Huge $C_{1,2}$] {};
        \node[state, tokens=4] at (11,12) (C21)  [label=above:\Huge $C_{2,1}$] {};
        \node[state, tokens=4] at (13,12) (C22)  [label=above:\Huge $C_{2,2}$] {};
        \node[state, draw = red] at (1,8) (7)  {\Huge $G^{\mathcal{A}}_{P_7}$};
        \node[state, draw = red] at (4,8) (5)  {\Huge $G^{\mathcal{A}}_{P_5}$};
        \node[state, draw = red] at (7,8) (4)  {\Huge $G^{\mathcal{A}}_{P_4}$};
        \node[state, draw = red] at (10,8) (1)  {\Huge $G^{\mathcal{A}}_{P_1}$};
        \node[state, draw = red] at (13,8) (3) {\Huge $G^{\mathcal{A}}_{P_3}$};
        \node[state, draw = red] at (16,8) (2)  {\Huge $G^{\mathcal{A}}_{P_2}$};
        \node[state, draw = red] at (19,8) (6)  {\Huge $G^{\mathcal{A}}_{P_6}$};
        
        \node[state, draw= white] at (7,1.5) (ena)  [label=right: \Huge $n^2\times Enc_A$] {};
        \node[state, draw= white] at (19.5,1.5) (enb)  [label=right: \Huge $n^2 \times Enc_B$] {};
        \node[state, draw= white] at (-1,10.5) (dec)  [label=right:\Huge $n^2 \times Dec$] {};

        \path[<-] (C11) edge  (4);
        \path[<-] (C11) edge  (1);
        \path[<-] (C11) edge  (7);
        \path[<-] (C11) edge  (5);
        
        \path[<-] (C12) edge  (3);
        \path[<-] (C12) edge  (5);
        
        \path[<-] (C21) edge  (2);
        \path[<-] (C21) edge  (4);
        
        \path[<-] (C22) edge  (6);
        \path[<-] (C22) edge  (1);
        \path[<-] (C22) edge  (2);
        \path[<-] (C22) edge  (3);
        
        \path[->] (A7) edge  (7);
        \path[->] (B7) edge  (7);
        \path[->] (A5) edge  (5);
        \path[->] (B5) edge  (5);
        \path[->] (A4) edge  (4);
        \path[->] (B4) edge  (4);
        \path[->] (A1) edge  (1);
        \path[->] (B1) edge  (1);
        \path[->] (A3) edge  (3);
        \path[->] (B3) edge  (3);
        \path[->] (A2) edge  (2);
        \path[->] (B2) edge  (2);
        \path[->] (A6) edge  (6);
        \path[->] (B6) edge  (6);
    \end{tikzpicture}
     }
     \caption{Recursive construction of $G^{\mathcal{A}}$. $A_{i,j}$, $B_{i,j}$ and $C_{i,j}$ denote the block-partition of $\mathbf{A}$, $\mathbf{B}$ and $\mathbf{C}$.}
     \label{fig:strarec}
  \end{subfigure}
  \caption{Blue vertices represent combinations of the input values from the factor matrices $\mathbf{A}$ and $\mathbf{B}$  used as input values for the seven sub-problems; red vertices represent the output of the seven sub-problems which are used to compute the values of the output matrix $\mathbf{C}$.}
\end{figure}
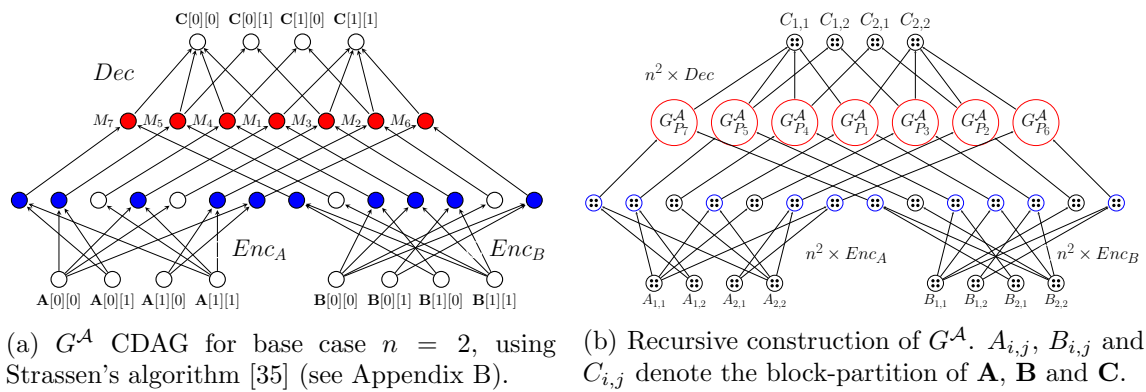 
\paragraph{Properties of $G^{\mathcal{A}}$:}
While the actual internal structure $G^{\mathcal{A}}$, and, in particular, the structure of encoder and decoder sub-CDAGs depends on the specific Strassen-like algorithm being used by $\mathcal{A}$, all versions share some properties of great importance.
Let $G(X,Y,E)$ denote an encoder CDAG for a fast multiplication algorithm $2\times 2$ base case, with $X$ (resp., $Y$) denoting the set of input (resp., output) vertices, and $E$ denoting the set of edges directed from $X$ to $Y$.

\begin{lemma}[{Lemma 3.3~\cite{ipdps2019}}]\label{lem:distinctencorder}
Let $G =(X, Y, E)$ denote an encoder graph for a fast matrix multiplication algorithm with base case $2\times 2$. There are no two vertices in Y with identical neighbors sets.
\end{lemma}
While the correctness of this Lemma can be simply verified by inspection in the case of Strassen's algorithm~\cite{strassen1969gaussian}, Lemma~\ref{lem:distinctencorder} generalizes the statement to \emph{all} encoders corresponding to fast matrix multiplication algorithms with base case $2\times 2$. From Lemma~\ref{lem:distinctencorder} we have:
\begin{lemma}\label{lem:distinct}
Let $\mathcal{A}\in \hmm{}$ and  let $P_1$ and $P_2$ be any two sub-problems generated by $\mathcal{A}$  with input size greater than $\nz{}\times \nz{}$, such that $P_2$ is not recursively generated while computing $P_1$ and vice versa. Then, the sub-CDAGs of $G^{\mathcal{A}}$ corresponding, respectively, to $P_1$ and to $P_2$ are vertex disjoint.
\end{lemma}

The following lemma, originally introduced for Strassen's algorithm in~\cite{bilardi2017complexity} and then generalized for Strassen-like algorithms with base case $2\times 2$ in in~\cite{ipdps2019}, captures a connectivity property of encoder sub-CDAGs. 
\begin{lemma}[{Lemma 3.1~\cite{ipdps2019}}]\label{lem:conneconder}
Given an encoder CDAG for any Strassen-like algorithm with base case $2\times 2$, for any subset $Y$ of its output vertices, there exists a subset $X$ of its input vertices, with $\min\{|Y|, 1 + \lceil\left(|Y|-1\right)/2\rceil \}\leq |X|\leq |Y|$, such that there exist $|X|$ vertex-disjoint paths connecting the vertices in $X$ to vertices in $Y$.
\end{lemma}
The proofs of Lemma~\ref{lem:distinctencorder} and Lemma~\ref{lem:conneconder} are based on an argument originally presented by Hopcroft and Kerr~\cite{hopcroft1971minimizing}. We refer the reader to~\cite{ipdps2019} for the proofs.
\section{Maximal sub-problems and their properties}\label{sec:msp}
For an algorithm $\mathcal{A}\in \hmm{}$, let $P'$ denote a sub-problem generated by $\mathcal{A}$. In our presentation we consider the entire matrix multiplication problem an \emph{improper} sup-problem generated by $\mathcal{A}$. We refer as the \emph{ancestor sub-problems} of $P'$ as the sequence of sub-problems $P'_0, P'_2,\ldots,P'_i$ generated by $\mathcal{A}$ such that $P'_{j+1}$ was recursively generated to compute $P'_j$ for $j=0,1,\ldots,i-1$, and such that $P$ was recursively generated to compute $P'_i$. Clearly, if $P'$ is the entire problem, then $P'$ has no ancestors.

Towards studying the \io{} complexity of algorithms in $\hmm{}$ we focus on the analysis of a particular set of sub-problems.


\begin{definition}[Maximal Sub-Problems (MSP)]\label{def:maxsubp}
Let $\mathcal{A} \in \hmm{}$ be an algorithm used to multiply matrices $\mathbf{A},\mathbf{B} \in \ri{}^{n\times n}$. If $n\leq 2\sqrt{M}$ we say that $\mathcal{A}$ does not generate any Maximal Sub-Problem (MSP). 
\item Let $P_i$ be a sub-problem generated by $\mathcal{A}$ with input size $n_i\times n_i$, with $n_i\geq 2M$, and such that all its ancestors sub-problems are computed, according to $\mathcal{A}$ using algorithms from the fast class. We say that:
\begin{itemize}
    \item $P_i$ is a Type 1 MSP of $\mathcal{A}$ if, according to $\mathcal{A}$, is computed using an algorithm from the standard class. If the entire problem is to be solved using an algorithm for the standard class, we say that the entire problem is the unique (improper) Type 1 MSP generated by $\mathcal{A}$.
    \item $P_i$ is a Type 2 MSP of $\mathcal{A}$ if, according to $\mathcal{A}$, is computed by  generating 7 sub-problems  according to the recursive scheme corresponding to an algorithm from the fast (Strassen-like) class, and if the generated sub-problems have input size strictly smaller than $2\sqrt{M}\times 2\sqrt{M}$. If the entire problem uses a recursive algorithm from the fast class to generate 7 sub-problems with input size smaller than $2\sqrt{M}\times 2\sqrt{M}$, we say that the entire problem is the unique, improper, Type 2 MSP generated by $\mathcal{A}$.
\end{itemize}
\end{definition}
In the following we denote as $\nu_1$ (resp., $\nu_2$) the number of Type 1 (resp., Type 2) MSPs generated by $\mathcal{A}$. 
Let $P_i$ denote the $i$-th MSP generated by $\mathcal{A}$ and let $G^{\mathcal{A}}_{P_i}$ denote the corresponding sub-CDAG of $G^{\mathcal{A}}$. We denote as $\mathbf{A}_i$ and $\mathbf{B}_i$ (resp., $\mathbf{C}_i$) the input factor matrices (resp., the output product matrix) of $P_i$.
\paragraph{Properties of MSPs and their corresponding sub-CDAGs:}
 By Definition~\ref{def:maxsubp}, we have that for each pair of distinct MMSPs $P_1$ and $P_2$, $P_2$ is not recursively generated by $\mathcal{A}$ in order to compute $P_1$ or vice versa. Hence, by Lemma~\ref{lem:distinct}, the sub-CDAGs of $G^\mathcal{A}$ that correspond each to one of the MSPs generated by $\mathcal{A}$ are vertex, disjoint.

In order to obtain our \io{} lower bound for algorithms in $\hmm{}$, we characterize properties regarding the minimum dominator size of an arbitrary subset of $\subpInput{}$ and $\subpOutput{}$.
\begin{lemma}[Proof in Appendix~\ref{app:domtype2}]\label{lem:domtype2}
Let $G^\mathcal{A}$ be the CDAG corresponding to an algorithm $\mathcal{A}\in\hmm{}$ which admits $n_1$ Type 1 MSPs. For each Type 1 MSP $P_i$ let $\subpInput{}_i$ denote the set of input vertices of the associated sub-CDAG $G^{\mathcal{A}}_{P_i}$ which correspond each to an entry of the input matrices $\mathbf{A}_i$ and $\mathbf{B}_i$. Further, we define $\subpInput{} = \cup_{i=1}^{\nu_1} \subpInput{}_i$.

Let  $Y\subseteq \subpInput{}$ in $G^\mathcal{A}$ such that $|Y \cap \subpInput{}_i|= a_i/\sqrt{b_i}$, with $a_i,b_i \in \mathbb{N}$,  $a_i \geq  b_i$ for $i=1,2,\ldots,\nu_1$, and such that $b_i=0$ if and only if $a_i=0$.\footnote{Here we use as convention that $0/0=0$.} Any dominator set $\dom{}$ of $Y$ satisfies $|\dom{}|\geq \min \{2M, \sum_{i=1}^{\nu_1} a_i / \sqrt{\sum_{i=1}^{\nu_1} b_i}\}$.
\end{lemma}

\begin{lemma}[Proof in Appendix~\ref{app:domtype1}]\label{lem:domtype1}
Let $G^\mathcal{A}$ be the CDAG corresponding to an algorithm $\mathcal{A}\in\hmm{}$ which admits $n_2$ Type 2 MSPs. Further let $\mathcal{Z}$ denote the set of vertices corresponding to the entries of the output matrices of the $n_2$ Type 2 MSPs.
Given any subset $Z\subseteq \subpOutput{}$ in $G^\mathcal{A}$ with $|Z|\leq 4M$, any dominator set $\dom$ of $Z$ satisfies $|\dom|\geq|Z|/2$.
\end{lemma}

For each Type 1 MSP $P_i$ generated by $\mathcal{A}$, with input size\footnote{In general, different Type 1 MSP may have different input sizes} $n_i\times n_i$, we denote as $T_i$ the set of variables whose value correspond to the $n_i^3$ elementary products  $\mathbf{A}_i[j][k]\mathbf{B}_i[k][j]$ for $j,k = 0,1,\ldots,n_i-1$. Further, we denote as $\mathcal{T}_i$ the set of vertices corresponding to the variables in $T_i$, and we define $\mathcal{T}= \cup_{i=1}^{\nu_1} \mathcal{T}_i$.
\begin{lemma}[Proof in Appendix~\ref{app:domtype3}]\label{lem:domtype3}
For any Type 1 MSPs generated by $\mathcal{A}$ consider $T'_i\subseteq T_i$.
Let $\subpInput{}_i^{(\mathbf{A})}\subseteq \subpInput{}_i$ (resp., $\subpInput{}_i^{(\mathbf{B})}\subseteq \subpInput{}_i$) denote a subset of the vertices corresponding to entries of $\mathbf{A}_i$ (resp., $\mathbf{B}_i$) which are multiplied in at least one of the elementary products in $T'_i$. Then any dominator $D$ of the vertices corresponding to $T'_i$ with respect to the the vertices in $\subpInput{}_i$ is such that
$$
    |D|\geq \max \{ |\subpInput{}_i'\cap A_i|,|\subpInput{}_i'\cap B_i|\}. 
$$
\end{lemma}

%% file: hybmatmul.tex
\begin{theorem} \label{thm:genmatmul}
Let $\mathcal{A}\in \hmm{}$ be an algorithm to multiply two square matrices $\mathbf{A},\mathbf{B} \in\ri^{n\times n}$.
If run on a sequential machine with  cache of size $M$ and such that up to $B$  memory words stored in consecutive memory locations can be moved from cache to slow memory  and vice versa using a single memory operation, $\mathcal{A}$'s \io{} complexity satisfies:
\begin{equation}\label{eq:hmmmain}
	IO_{\mathcal{A}}\left(n,M,B\right) \geq 
	\max \{2n^2,c|\mathcal{T}|M^{-1/2},\nu_2 M\}B^{-1}
\end{equation}
for $c=0.38988157484$, where $|\mathcal{T}|$ denotes the total number of internal elementary products computed by the Type 1 MSPs  generated  by $\mathcal{A}$ and $\nu_2$ denotes the total number of Type 2 MSPs generated by $\mathcal{A}$.

If run on $\nproc{}$ processors each equipped with a local memory of size $M < n^2$ and where for each \io{} operation it is possible to move up to $\msgsize{}$ words, $\mathcal{A}$'s \io{} complexity satisfies:
\begin{equation}\label{eq:hmmpar}
	IO_{\mathcal{A}}\left(n,M,\msgsize{},\nproc{}\right) \geq 
	\max \{c|\mathcal{T}|M^{-1/2},\nu_2 M\}\left(\nproc{}\msgsize{}\right)^{-1}
\end{equation}
\end{theorem}
\begin{proof}
We prove the result in~\eqref{eq:hmmmain} (resp.,~\eqref{eq:hmmpar}) for the case $B=1$ (resp., $\msgsize{}=1$). The result then trivially generalizes for a generic $B$ (resp., $\msgsize{}$). We first prove the result for the sequential case in in~\eqref{eq:hmmmain}. The bound for the parallel case in~\eqref{eq:hmmpar} will be obtained as a simple generalization. 

The fact that $IO_{\alg}(n,M,1)\geq 2n^2$ follows trivially from the fact that as in our model the input matrices $\mathbf{A}$ and $\mathbf{B}$ are initially stored in slow memory, it will necessary to move the entire input to the cache at least once using at least $2n^2$ \io{} operations.  If $\alg{}$ does not generate any MSPs the statement in~\eqref{eq:hmmmain} is trivially verified. In the following, we assume $\nu_1+\nu_2\geq 1$. 

Let $G^{\mathcal{A}}$ denote the CDAG associated with algorithm $\mathcal{A}$ according to the construction in Section~\ref{sec:CDAG}. 
By definition, and from Lemma~\ref{lem:distinct}, the $\nu_1+\nu_2$ sub-CDAGs of $G^{\mathcal{A}}$ corresponding each to one of the MSPs generated by $\mathcal{A}$  are vertex-disjoint. Hence,  the $\mathcal{T}_i$'s are a partition of $\mathcal{T}$ and $|\mathcal{T}|=\sum_{i=1}^{\nu_1} |\mathcal{T}_i|$.

By Definition~\ref{def:maxsubp}, the  MSP generated by $\alg{}$ have input (resp., output) matrices of size greater or equal to $2\sqrt{M} \times 2\sqrt{M}$. Recall that we denote as $\subpOutput{}$ the set of vertices which correspond to the outputs of the $\nu_2$ Type 2 MSPs, we have $|\subpOutput{}|\geq 4M\nu_l$. 


Let $\mathcal{C}$ be any computation schedule for the sequential execution of $\mathcal{A}$ using a cache of size $M$. We partition $\mathcal{C}$ into non-overlapping segments  $\mathcal{C}_1,\mathcal{C}_2,\ldots$ such that during each $\mathcal{C}_j$ either (a) exactly $M^{3/2}$ distinct values~\footnote{For simplicity of presentation, we assume $M^{3/2}\in \mathbb{N}^{+}$.} corresponding to vertices in $\mathcal{T}$, denoted as $\mathcal{T}^{(j)}$, are explicitly computed (i.e., not loaded from slow memory), 
or (b) $4M$ distinct values corresponding to vertices in $\subpOutput{}$ (denoted as $\subpOutput{}_j$) are evaluated for the \emph{first time}. Clearly there are at least $\max \{|\mathcal{T}|/{M^{3/2}},\nu_2 \}$ such segments. 

Below we show that the number $g_j$ of  \io{} operations  executed during each $\mathcal{C}_j$ satisfies $g_j\geq c M$ for case (a) and $g_j\geq M$ for case (b), from which the theorem follows. 

\textbf{Case (a):} 
For each Type 1 MSP $P_i$  let $\mathcal{T}^{(j)}_i = \mathcal{T}^{(j)}\cap \mathcal{T}_i$. As the $\nu_1$ sub-CDAGs corresponding each to one of the Type 1 MSPs are vertex-disjoint, so are the sets $\mathcal{T}_i$. Hence, the $\mathcal{T}^{(j)}_i$'s constitute a partition of $\mathcal{T}^{(j)}$. 
    Let $\mathbf{A}_i$ and $\mathbf{B}_i$ (resp., $\mathbf{C}_i$) denote the input matrices (resp., output matrix) of $P_i$ with $\mathbf{A}_i,\mathbf{B}_i,\mathbf{C}_i\in \ri^{n_i \times n_i}$, and let $A_i$ and $B_i$ (resp., $C_i$) denote the set of the variables corresponding to the entries of $\mathbf{A}_i$ and $\mathbf{B}_i$ (resp., $\mathbf{C}_i$). Further, we denote as $T_i$ the set of values corresponding to the vertices in $\mathcal{T}_i$.
    For $r,s = 0,1,\ldots,n_i-1$, we say that $\mathbf{C}_i[r][s]$ is ``\emph{active during} $\mathcal{C}_j$'' if \emph{any} of the elementary multiplications $\mathbf{A}_i[r][k]\mathbf{B}_i[k][s]$, for $k= 0,1,\ldots,n_i-1$, correspond to any of the vertices in $\mathcal{T}^{(j)}_i$. Further we say that a $\mathbf{A}_i[r][s]$ (resp., $\mathbf{B}_i[r][s]$) is ``\emph{accessed during} $\mathcal{C}_j$'' if \emph{any} of the elementary multiplications $\mathbf{A}_i[r][s]\mathbf{B}_i[s][k]$ (resp., $\mathbf{A}_i[k][r]\mathbf{B}_i[r][s]$), for $k= 0,1,\ldots,n_i-1$, correspond to any of the vertices in $\mathcal{T}^{(j)}_i$.
    
    Our analysis makes use of the following property of standard matrix multiplication algorithms:
    \begin{lemma}[{Loomis-Whitney inequality~\cite[Lemma 2.2]{irony2004communication}}]\label{lem:mywl}
    Let $\subpInput{}_{i,\mathbf{A}}'$ (resp., $\subpInput{}_{i,\mathbf{B}}'$) denote the set of vertices corresponding to the entries of $\mathbf{A}_i$ (resp., $\mathbf{B}_i$) which are accessed during $\mathcal{C}_j$, and let $\subpOutput{}
    _i'$ denote the set of vertices corresponding to the entries of $\mathbf{C}_i$  which are active during $\mathcal{C}_j$. Then 
    \begin{equation}\label{eq:mywl}
        |\mathcal{T}^{(j)}_i|\leq \sqrt{|\subpInput{}_{i,\mathbf{A}}'||\subpInput{}_{i,\mathbf{B}}'||\mathcal{Z}_i'|}. 
    \end{equation}
    \end{lemma}
    
    Lemma~\ref{lem:mywl} is a reworked version of a property originally presented by Irony et al.~\cite[Lemma 2.2]{irony2004communication}, which itself is a consequence of the Loomis-Whitney geometric theorem~\cite{loomis1949}.
    
  Let $\mathbf{C}_i[r][s]$ be active during $\mathcal{C}_j$. In order to compute $\mathbf{C}_i[r][s]$ \emph{entirely} during $\mathcal{C}_j$ (i.e., without using partial accumulation of the summation $\sum_{k=0}^{n_i-1} \mathbf{A}_i[r][k]\mathbf{B}[k][s]$), it will be necessary to evaluate all the $n_i$ elementary products $\mathbf{A}_i[r][k]\mathbf{B}_i[k][s]$, for $k=0,1,\ldots, n_i-1$, during $\mathcal{C}_j$ itself. Thus, at most $\lfloor|\mathcal{T}^{(j)}_i|/{n_i}\rfloor$ entries of $\mathbf{C}_i[r][s]$ can be entirely computed during $\mathcal{C}_j$.
   
  Let $\mathbf{C}_i[r][s]$ denote an entry of $\mathbf{C}_i$ which is active but not entirely computed during $\mathcal{C}_j$. There are two possible scenarios:
  \begin{itemize}
      \item $\mathbf{C}_i[r][s]$ is computed during $\mathcal{C}_j$: The computation thus requires for a partial accumulation of $\sum_{k=0}^{n_i-1} \mathbf{A}_i[r][k]\mathbf{B}[k][s]$  to have been previously computed and either held in the cache at the beginning of $\mathcal{C}_j$, or to moved to cache using a \texttt{read} \io{} operation during $\mathcal{C}_j$; 
      \item $\mathbf{C}_i[r][s]$ is not computed during $\mathcal{C}_j$: As $\mathcal{C}$ is a parsimonious computation, the partial accumulation of $\sum_{k=0}^{n_i-1} \mathbf{A}_i[r][k]\mathbf{B}[k][s]$ obtained from the elementary products computed during $\mathcal{C}_j$ must either remain in the cache at the end of $\mathcal{C}_j$, or be moved to slow memory using a \texttt{write} \io{} operation during $\mathcal{C}_j$;
  \end{itemize}
  In both cases, any  partial accumulation  either held in memory at the beginning (resp., end) of $\mathcal{C}_j$ or read from slow memory to cache (resp., written from cache to slow memory) during $\mathcal{C}_j$ is, by definition, not shared between multiple entries in $\mathbf{C}_i$.
   
  Let $G^{\mathcal{A}}_{P_i}$ denote the sub-CDAG of $G^{\mathcal{A}}$ corresponding to the Type 1 MSP $P_i$. In the following, we refer as $D'_i$ to the set of vertices of $G^{\mathcal{A}}_{P_i}$ corresponding to the values of such partial accumulators.
  For each of the least $|\dom{}'_i|= \max\{0,|\mathcal{Z}_i'|-|\mathcal{T}^{(j)}_i|/n_i\}$  entries of $\mathbf{C}_i$ which are active but not entirely computed during $\mathcal{C}_j$, either one of the entries of the cache must be occupied at the beginning of $\mathcal{C}_j$, or one \io{} operation is executed during  $\mathcal{C}_j$. Let $\dom{}' = \cup_{i=1}^{\nu_1}\dom{}'_i$. As, by Lemma~\ref{lem:distinct}, the sub-CDAGs corresponding to the $\nu_1$ Type 1 MSPs are vertex disjoint, so are the the sets $\dom{}'_i$. Let $\mathcal{Z}'=\sum_{i=1}^{\nu_l} |\mathcal{Z}_i'|$. We have:
  \begin{equation}\label{eq:dom1}
      |\dom{}'| = \sum_{i=1}^{\nu_1}|\dom{}'_i| = \sum_{i=1}^{\nu_1} \max\{0,|\mathcal{Z}_i'|-|\mathcal{T}^{(j)}_i|/n_i\} \geq |\mathcal{Z}|-|\mathcal{T}^{(j)}|/2\sqrt{M},
  \end{equation}
  where the last passage follows from the fact that, by Definition~\ref{def:maxsubp}, $n_i\geq2M$. 
   
  From Lemma~\ref{lem:mywl}, the set of vertices $\subpInput{}_{i,\mathbf{A}}'$ (resp., $\subpInput{}_{i,\mathbf{B}}'$) which correspond to entries of $\mathbf{A}_i$ (resp., $\mathbf{B}_i$) which are accessed during $\mathcal{C}_j$ satisfies
  $
      |\subpInput{}_{i,\mathbf{A}}'||\subpInput{}_{i,\mathbf{B}}'| \geq  |\mathcal{T}^{(j)}_i|^2/|\mathcal{Z}_i'|
  $.
  Hence, at least $|\subpInput{}_{i,\mathbf{A}}'|+|\subpInput{}_{i,\mathbf{B}}'| \geq 2 |\mathcal{T}^{(j)}_i|/\sqrt{|\mathcal{Z}_i'|}$ entries from the input matrices of $P_i$ are accessed during $\mathcal{C}_j$. Let $\mathcal{Y}$ denote the set of vertices corresponding to the entries of the input matrices $\mathbf{A}_i,\mathbf{B}_i$ of $P_i$. From Lemma~\ref{lem:domtype3} we have that there exists a set $\subpInput{}_i'\subseteq \subpInput{}_i$, with\sloppy $|\subpInput{}_i'|\geq \max \{|\subpInput{}_{i,\mathbf{A}}'|,|\subpInput{}_{i,\mathbf{B}}'|\}\geq |\mathcal{T}^{(j)}_i|/\sqrt{|\mathcal{Z}_i'|}$,
  such that the vertices in $\subpInput{}_i'$ are connected by vertex disjoint pats to the vertices in $\mathcal{T}^{(j)}_i$. Let $Y = \cup_{i=1}^{\nu_1}\subpInput{}_i'$. As, by Lemma~\ref{lem:distinct}, the sub-CDAGs corresponding to the $\nu_l$ Type 1 MSPs are vertex disjoint, so are the the sets $\subpInput{}_i'$ for $i=1,2,\ldots,\nu_1$. Hence
    $$
        |Y| = \sum_{i=1}^{\nu_1}|\subpInput{}_i'|\geq  \sum_{i=1}^{\nu_1}\frac{|\mathcal{T}^{(j)}_i|}{\sqrt{|\mathcal{Z}_i'|}}.
    $$
    From Lemma~\ref{lem:domtype2} any dominator $\dom{}_{Y}$ of $Y$, must be such that $$|\dom{}_{Y}|\geq \min\bigl\{2M, \frac{\sum_{i=1}^{\nu_1} |\mathcal{T}^{(j)}_i| }{\sum_{i=1}^{\nu_1}\sqrt{|\mathcal{Z}_i'|}}\bigr\} = \min\bigl\{2M,\frac{ |\mathcal{T}^{(j)}| }{\sqrt{|\mathcal{Z}'|}}\bigr\}.$$ 
    Hence, we can conclude that any dominator $\dom{}''$ of $\mathcal{T}^{(j)}$ must be such that 
    \begin{equation}\label{eq:dom2}
        |\dom{}''|\geq \min\bigl\{2M,|\mathcal{T}^{(j)}|/\sqrt{|\mathcal{Z}'|}\bigr\}.
    \end{equation}
    
    Consider the set $\dom{}$ of vertices of $G_{\mathcal{A}}$ corresponding to the at most $M$ values stored in the cache at the beginning of $\mathcal{C}_j$ and to the at most $g_j$ values loaded into the cache form the slow memory (resp., written into the slow memory from the cache) during  $\mathcal{C}_j$ by means of a \texttt{read} (resp., \texttt{write}) \io{} operation. Clearly, $|\dom|\leq M +g_j$.
    
    In order for the $M^{3/2}$ values from $\mathcal{T}^{(j)}$ to be computed during $\mathcal{C}_j$ there must be no path connecting any vertex in $\mathcal{T}^{(j)}$, and, hence, $Y$, to any input vertex of $G^{\mathcal{A}}$ which does not have at least one vertex in $\dom{}$, that is  $\dom{}$ has to admit a subset $\dom{}''\subseteq \dom{}$ such that $\dom{}''$ is a \emph{dominator set} of $\mathcal{T}^{(j)}$. Note that, as the values corresponding to vertices in $\mathcal{T}^{(j)}$ are actually computed during $C_J$ (i.e., not loaded from memory using a $\texttt{read}$ \io{} operation), $\dom{}''$ does not include vertices in $\mathcal{T}^{(j)}$ itself. Further, as motivated in the previous discussion, $\dom{}$ must include all the vertices in the set $\dom{}'$ corresponding to values of partial accumulators of the active output values of Type 1 MSPs during $\mathcal{C}_j$. 
    
	By construction, $\dom{}'$ and  $\dom{}''$ are vertex disjoint. Hence, from~\eqref{eq:dom1} and~\eqref{eq:dom2} we have:
	\begin{equation*}
	    |\dom{}| \geq |\dom{}'|+|\dom{}''| \geq |\mathcal{Z}'|-|\mathcal{T}^{(j)}|/ 2\sqrt{M}+ \min\bigl\{2M,|\mathcal{T}^{(j)}|/\sqrt{|\mathcal{Z}'|}\bigr\}.
	\end{equation*}
	As, by construction, $|\mathcal{T}^{(j)}|=M^{3/2}$, we have:
	\begin{align}\label{eq:last1}
	    |\dom{}| &> |\mathcal{Z}'| - M/2 + \min\bigl\{2M,M^{3/2}/\sqrt{|\mathcal{Z}'|}\bigr\}.
	\end{align}
	By studying its derivative after opportunely accounting for the minimum, we have that~\eqref{eq:last1}
	 is minimized for $|\mathcal{Z}'| = 2^{-2/3}M$. Hence we have:
	 $|\dom{}| > 2^{-2/3}M + 2^{1/3}M^{3/2} - M/2 = 1.38988157484 M$.
	Whence $|D|\leq M +g_j$, which implies $g_j\geq |D|-M > 0.38988157484 M$, as stated above.\\
	
	\textbf{Case (b):} In order for the $4M$ values from $\subpOutput{}_j$ to be computed during $\mathcal{C}_j$ there must be no path connecting any vertex in $\subpOutput{}_j$ to any input vertex of $G_{\mathcal{A}}$ which does not have at least one vertex in $\dom_j$, that is  $\dom_j$ has to be a \emph{dominator set} of $\subpOutput{}_j$.
	 From Lemma~\ref{lem:domtype1}, any dominator set $D$ of any subset $Z\subseteq \subpOutput{}$ with $|Z|\leq 4M$ satisfies $|D|\geq |Z|/2$, whence $M+g_i\geq |\dom_i|\geq |\subpOutput{}_j|/2 =2M$, which implies $g_j\geq M$ as stated above. This concludes the proof for the sequential case in~\eqref{eq:hmmmain}.\\
  
	
	The proof for the bound for the parallel model in ~\eqref{eq:hmmpar}, follows from the observation that at least one of the $\nproc{}$ processors, denoted as $\nproc{}^*$, must compute at least $|\mathcal{T}|/\nproc{}$ values corresponding to vertices in $\mathcal{T}$ or $|\subpOutput{}|/\nproc{}$ values corresponding to vertices in $\subpOutput{}$ (or both). The bound follows by applying the same argument discussed for the sequential case to the computation executed by $\nproc^*$. 
\end{proof}

Note that if $\mathcal{A}$ is such that the product is entirely computed using an algorithm from the standard class (resp., a fast matrix multiplication algorithm), the bounds of Theorem~\ref{thm:genmatmul} corresponds asymptotically to the results of~\cite{jia1981complexity} for the sequential case and~\cite{irony2004communication} for the parallel case (resp., the results in~\cite{bilardi2017complexity}).\\

 The bound in~\eqref{eq:hmmpar} can be further generalized to a slightly different model in which each of the $\nproc{}$ processors is equipped with a cache memory of size $M$ and a slow memory of unbounded size. In such case, the \io{} complexity of the algorithms in $\hmm{}$ corresponds to the total number of both messages received and the number of \io{} operations used to move data from cache to slow memory and vice versa (i.e., \texttt{read} and \texttt{write}) executed by at least one of the  $\nproc{}$ processors.

\paragraph{\io{} lower bound for uniform, non stationary algorithms in $\uhmm{}$:}
For the sub-class of uniform, non stationary algorithms $\uhmm{}$, given the values of $n$, $M$ and $\nz{}$ is possible to compute a closed form expression for the values of $\nu_1,\nu_2$ and $|\mathcal{T}|$.\footnote{The constants terms in Theorem~\ref{thm:corgenmatmul} assume that $n,\nz{}$ and $M$ are powers of two. If that not the case the statement holds with minor adjustments.} Then, by applying  Theorem~\ref{thm:genmatmul} we have:

\begin{theorem}\label{thm:corgenmatmul}
Let $\mathcal{A}\in \uhmm{}$ be an algorithm to multiply two square matrices $\mathbf{A},\mathbf{B} \in\ri^{n\times n}$. If run on a sequential machine with  cache of size $M$ and such that up to $B$  memory words stored in consecutive memory locations can be moved from cache to slow memory  and vice versa using a single memory operation, $\mathcal{A}$'s \io{} complexity satisfies:
\begin{equation}\label{eq:hmmmaincor}
	IO_{\mathcal{A}}\left(n,M,B\right) \geq \max \{2n^2, \left(\frac{n}{\max\{\nz{},2\sqrt{M}\}}\right)^{\log_2 7} \left(\max\bigl\{1, \frac{\nz{}}{2\sqrt{M}}\bigr\}\right)^3 M\}B^{-1}
\end{equation}
 If run on $\nproc{}$ processors each equipped with a local memory of size $M < n^2$ and where for each \io{} operation it is possible to move up to $\msgsize{}$ memory words, $\mathcal{A}$'s \io{} complexity satisfies:
\begin{equation}\label{eq:hmmmaincorpar}
	IO_{\mathcal{A}}\left(n,M,\msgsize{},P\right) \geq \left(\frac{n}{\max\{\nz{},2\sqrt{M}\}}\right)^{\log_2 7} \left(\max\Bigl\{1, \frac{\nz{}}{2\sqrt{M}}\Bigr\}\right)^3 \frac{M}{\nproc{}\msgsize{}}.
\end{equation}
\end{theorem}
\begin{proof}
The statement follows by bounding the values $\nu_1$,$\nu_2$ and $|\mathcal{T}|$ for $\mathcal{A}\in\uhmm{}$, and by applying the general result from Theorem~\ref{thm:genmatmul}.
In order to simplify the presentation, in the following we assume that the values $n,\nz{},M$ are powers of two. If that is not the case, the theorem holds with some minor adjustments to the constant  multiplicative factor. 

Let let $i$ be the smallest value in $\mathbb{N}$ such that $n/2^i = \max\{\nz{}, 2\sqrt{M}\}$. By definition of $\hmm{}$, at each of the $i$ recursive levels $\mathcal{A}$ generates $7^i$ sub-problems of size $n/2^i$.  

\begin{itemize}
    \item If $\nz{}> 2\sqrt{M}$, $\mathcal{A}$ generates 
    $$\nu_1 = 7^i = 7^{\log_2 n/\nz{}}= \left(\frac{n}{\nz{}}\right)^{\log_2 7}$$
    Type 1 MSP each with input size $\nz{}\times \nz{}$. As, by Definition~\ref{def:maxsubp}, the Type 1 MSP are input dijoint we have:
    $$|\mathcal{T}| = \nu_1 \nz{}^3 = \left(\frac{n}{\nz{}}\right)^{\log_2 7} \nz{}^3.$$ 
    \item Otherwise, if $\nz{}\leq 2\sqrt{M}$, $\mathcal{A}$ generates 
    $$\nu_2 = 7^i = 7^{\log_2 n/2\sqrt{M}}= \left(\frac{n}{2\sqrt{M}}\right)^{\log_2 7}$$
    Type 2 MSP each with input size $2\sqrt{M} \times 2\sqrt{M}$. 
\end{itemize}
The statement then follows by applying the result in Theorem~\ref{thm:genmatmul}.
\end{proof}
 
 Theorem~\ref{thm:corgenmatmul} extends the result by Scott~\cite{scott2015complexity} by expanding the class of hybrid matrix multiplication algorithms being considered (e.g., it does not limit the class of
 standard matrix multiplication to the divide and conquer algorithm based on block-partitioning), and by removing the assumption that no intermediate value may be recomputed. 
\paragraph*{On the tightness of the bound:} An opportune composition of the cache-optimal version of the Strassen's algorithm~\cite{strassen1969gaussian} (as discussed in ~\cite{ballard2012communicationalg}) with the standard cache-optimal divide and conquer algorithm for square matrix multiplication based on block-partitioning~\cite{cannon1969cellular} leads to a sequential hybrid algorithms  in $\hmm{}$ (resp., $\uhmm{}$) whose \io{} cost asymptotically matches the \io{} complexity lower bounds in Theorem~\ref{thm:genmatmul}~\eqref{eq:hmmmain} (resp., Theorem~\ref{thm:corgenmatmul}~\eqref{eq:hmmmaincor}). 

Parallel algorithms in $\hmm{}$ (resp., $\uhmm{}$) asymptotically matching the \io{} lower bounds in for the parallel case  in Theorem~\ref{thm:genmatmul}~\eqref{eq:hmmpar} (resp., Theorem~\ref{thm:corgenmatmul}~\eqref{eq:hmmmaincorpar}) can be obtained by composing the communication avoiding version of Strassen's algorithm by Ballard et al.~\cite{ballard2012communicationalg} with the communication avoiding ``\emph{2.5}'' standard algorithm by Solomonik and Demmel~\cite{solomonik2011communication}. 

Hence, the lower bounds in Theorem~\ref{thm:genmatmul} and  Theorem~\ref{thm:corgenmatmul}  are asymptotically tight and the mentioned algorithms form $\hmm{}$ and $\uhmm{}$ whose \io{} cost asymptotically match the lower bounds are indeed \io{} optimal.

Further, as the mentioned \io{} optimal algorithms from $\hmm{}$ and $\uhmm{}$ do not recompute any intermediate value~\footnote{We do not consider replication of the input used by the mentioned parallel algorithms as recomputation, but rather as a \emph{repeated access} to the input values.}, we can conclude that using recomputation may lead to at most a constant factor reduction of the \io{} cost of hybrid algorithms in $\hmm{}$ and $\uhmm{}$.
\paragraph*{Generalization to fast matrix multiplication model with base other than $2\times 2$:}
The general statement of Theorem~\ref{thm:genmatmul} can be extended to by enriching $\hmm{}$ to include any fast Strassen-like algorithm with base case other than $2\times 2$ provided that the associated encoder CDAG satisfies properties equivalent to those expressed by Lemma~\ref{lem:distinct}(i.e., the input disjointedness of the sub-problems generated at each recursive step) and Lemma~\ref{lem:conneconder} (i.e., the connectivity between input and output of the encoder CDAGs via vertex disjoint paths) for the $2\times 2$ base. If these properties hold, so does the general structure of Theorem~\ref{thm:genmatmul}, given an opportune adjustment of the definition of maximal sub-problem. 

%% file: conclusion.tex
\vspace{-2mm}
This work introduced the first characterization of the \io{}
complexity of hybrid matrix multiplication algorithms combining fast Strassen-like algorithms with standard algorithms.  We established asymptotically
tight lower bounds that hold even when recomputation is allowed.
The generality of the technique used the analysis  makes it promising for the analysis of other hybrid
recursive algorithms, e.g., for hybrid algorithms for integer multiplication~\cite{BilardiS19}.


Our results contribute to the study of the effect of recomputation with respect to the \io{} complexity of  CDAG algorithms. While we are far from a characterization of those CDAGs for which recomputation is effective, this broad goal remains a fundamental challenge for any attempt toward a general theory of the communication requirements of computations.

%% file: ack.tex
The author would like to thank Gianfranco Bilardi at the University of Padova for the helpful suggestions and discussions.

%% file: nappendix.tex
\section{Proofs of technical lemmas}\label{app:proofs}
\subsection{Proof of Lemma~\ref{lem:domtype2}}\label{app:domtype2}
The proof of Lemma~\ref{lem:domtype2} uses an analysis similar to that in~\cite{bilardi2017complexity}(Lemma 6), albeit with several important variations.
Before delving into the details of the proof of Lemma~~\ref{lem:domtype2} we introduce the following technical lemma:

\begin{lemma}\label{lem:lastminute}
Let $a_i,a_2,\ldots a_i\in \mathbb{N}$ (resp., $b_i,b_2,\ldots b_i\in \mathbb{N}$) such that $b_j = 0$ if and only if $a_i=0$ and such that:
$$
\frac{a_1}{\sqrt{b_1}}\geq \frac{a_2}{\sqrt{b_2}}\geq ....\geq \frac{a_i}{\sqrt{b_i}}
$$
using the convention that $0/0=0$. Then we have:
\begin{equation}
    \frac{a_1}{\sqrt{b_1}}+ \frac{1}{2}\sum_{j=2}^i \frac{a_j}{\sqrt{b_j}} \geq \frac{\sum_{j=1}^i a_i}{\sqrt{\sum_{j=1}^j b_j}}
\end{equation}
\end{lemma}
\begin{proof}
The proof is by induction on the value of $i$. In the base case case $i=1$, hence $a=a_1$ and $b=b_1$. Thus the statement is trivially verified. 
We assume inductively that the statement holds for $i\geq 1$ and we show that it holds also for $i+1$.
Let 
$$
\frac{a_1}{\sqrt{b_1}}\geq \frac{a_2}{\sqrt{b_2}}\geq ....\geq \frac{a_i+1}{\sqrt{b_i+1}}
$$
be can therefore apply the inductive hypothesis to the first $i$ elements and obtain:
\begin{equation}\label{eq:lastmin1}
    \frac{a_1}{\sqrt{b_1}}+ \frac{1}{2}\sum_{j=2}^{i+1} \frac{a_j}{\sqrt{b_j}} \geq 
    \frac{\sum_{j=1}^i a_j}{\sqrt{\sum_{j=1}^i b_j}}+ \frac{1}{2}\frac{a_{i+1}}{\sqrt{b_{i+1}}}.
\end{equation}
Let $a'= \sum_{j=1}^i a_j$ and $b'= \sum_{j=1}^i b_j$. If $a_{i+1} = 0$ the statement is trivially verified. In the following we consider the case for which  $a_{i+1} > 0$ and, by assumption, $b_{i+1}>0$.

From~\eqref{eq:lastmin1}, we have that the lemma is guaranteed to be verified if the following holds:
\begin{equation}\label{eq:lastmin2} 
    \frac{a'}{\sqrt{b'}}+ \frac{1}{2}\frac{a_{i+1}}{\sqrt{b_{i+1}}} \geq \frac{a'+a_{i+1}}{\sqrt{b'+b_{i+1}}} = \frac{a'}{\sqrt{b'+b_{i+1}}} + \frac{a_{i+1}}{\sqrt{b'+b_{i+1}}}
\end{equation}
Clearly~\eqref{eq:lastmin2} is verified for $b'\geq b_{i+1}$. Since, by assumption, $b=b'+b_{i+1}$, we can therefore conclude that the statement is verified for $b'\geq3b/4$.

In the following we consider the case for $b'<3b/4$.
Let $a' = ax$ (resp., $b'=by$), with $x\in (0,1]$. By assumption $x,y\neq 0$. Further, as we are considering the case such that $b'<3b/4$, we have that $y\in (0,3/4)$. Since, by assumption, $a = a'+a_{i+1}$ and $b = b'+b_{i+1}$, we have:
\begin{align*}
    \frac{a'}{\sqrt{b'}}+ \frac{1}{2}\frac{a_{i+1}}{\sqrt{b_{i+1}}} & = 
    \frac{ax}{\sqrt{by}}+ \frac{1}{2}\frac{a(1-x)}{\sqrt{b(1-y)}}\\
    &=\frac{a}{\sqrt{b}}\left(\frac{x}{\sqrt{y}}+ \frac{1-x}{2\sqrt{1-y}}
    \right)
\end{align*}
Therefore, the Lemma is clearly verified if the following holds:
\begin{equation}\label{eq:lastmin3}
    \frac{x}{\sqrt{y}}+ \frac{1-x}{\sqrt{1-y}} = \frac{2x\sqrt{1-y}+ (1-x)\sqrt{y}}{2\sqrt{y}\sqrt{1-y}} \geq 1.
\end{equation}
By multiplying both sides of the previous inequality by $2\sqrt{y}\sqrt{y-1}$ (which is clearly greater than $0$), we have that~\eqref{eq:lastmin3}  holds if:
\begin{equation}\label{eq:lastmin4}
    2x\sqrt{1-y}+ \left(1-x\right)\sqrt{y}\geq 2\sqrt{y}\sqrt{1-y},
\end{equation}
which in turn holds if:
\begin{equation}\label{eq:lastmin4}
    2\sqrt{1-y}\left(x-\sqrt{y}\right)\geq  \sqrt{y}\left(x-1\right),
\end{equation}
As $y\in (0,3/4)$ we have $2\sqrt{1-y}>1>\sqrt{y}$ and $\left(x-\sqrt{y}\right)\left(x-1\right)$. Hence~\eqref{eq:lastmin4} holds ant the Lemma follows.
\end{proof}

Recall that given an algorithm $\mathcal{A}\in \hmm{}$ used to multiply input squared matrices $\mathbf{A}, \mathbf{B} \in \ri^{n \times n}$, we denote as  $G^{\mathcal{A}}$  the corresponding CDAG constructed according to the description in Section~\ref{sec:CDAG}. Further, we denote as $\globalInput{}$ the set of input vertices of $G^{\mathcal{A}}$. That is the set of vertices corresponding each to the entries in the input matrices $\mathbf{A}$ and $\mathbf{B}$. We now present the proof of Lemma~\ref{lem:domtype2}.

\begin{proof}[Proof of Lemma~\ref{lem:domtype2}]
The proof proceeds by induction on the number of Type 1 MSPs $\nu_1$ generated by $\mathcal{A}$. In the base case $\nu_1=1$ and, by Definition~\ref{def:maxsubp} the entire problem is the only, improper Type 1 MSP generated by $\mathcal{A}$. Hence, the sets
$\subpInput{}$ and $\globalInput{}$ coincide and the statement is trivially verified.
  
Assuming now inductively that the statement holds $\nu_1=k\geq 1$,  we shall show
it also holds for $\nu_1 = k+1$. 

As $\nu_1>1$, we have that the algorithm executes at least one recursive step. Let $P^{(j)}$ denote the 7 sub-problems generated at the first recursion step. We distinguish two cases:
\textbf{(a)} at least two of of the seven sub-problems $P^{(j)}$ generate each at least one Type 1 MSP; \textbf{(b)} only one of the seven sub-problems generates all $\nu_1$ Type 1 MSPs. We first address case \textbf{(a)}, as case \textbf{(b)} will follow from a simple extension.

For case \textbf{(a)}, let $G^{(j)}$, for $j=1,2,\ldots,7$ denote the seven sub-CDAGs of $G^{\mathcal{A}}$, each corresponding to one of the seven sub-problems generated in the first recursive step of $\mathcal{A}$ according to the chosen Strassen-like scheme as discussed in Section \ref{sec:algdef}. By Definition~\ref{def:maxsubp} and as, by assumption, $\mathcal{A}$ generates Type 1 MSPs, we have that each of these seven sub-problems has input size greater or equal to  $2\sqrt{M} \times 2\sqrt{M}$.  Further, each of the seven sub-problems $P^{(j)}$ generates at most $\nu_1-1$ Type 1 MSPs.

Let $Y^{(j)}$ (resp., $\subpInput{}^{(j)}$) denote the subsets of $Y$ (resp., $\subpInput{}$) in $G^{(j)}$, for $j=1,2,\ldots,7$. That is $Y^{(j)} = Y\cap \mathcal{Y}^{(j)}$, By  Lemma~\ref{lem:distinct}  the $G^{(j)}$'s have distinct input values and, hence, are pairwise vertex-disjoint sub-CDAGs of $G^{\mathcal{A}}$. Thus, the $Y^{(1)},Y^{(2)},\ldots,Y^{(7)}$ partition $Y$ (resp., $\subpInput{}^{(1)},\subpInput{}^{(2)},\ldots,\subpInput{}^{(7)}$ partition $\subpInput{}$)  and $\sum_{j=1}^7 |Y^{(j)}| = |Y|$ (resp., $\sum_{j=1}^7 |\subpInput{}^{(j)}| = |\subpInput{}|$).  Let $|Y^{(j)}|=a_j / \sqrt{b_j}$, we have $|Y|= \sum_{j=1}^7 a_j / \sqrt{b_j}$.

    By the inductive hypothesis, any dominator set $D^{(j)}$ of $Y^{(j)}$ with respect to the set $\mathcal{K}^{(j)}$ composed by the input vertices of $G^{(j)}$ must be such that $|D^{(j)}|\geq \min\{2M, a_j/\sqrt{b}_j\}$. By Definition~\ref{def:dominator} this implies that vertices in  $Y^{(j)}$ can be connected to a subset $K^{(j)}\subseteq \mathcal{K}^{(j)}$ of the input vertices of $G^{(j)}$  such that $|K^{(j)}|\geq  a_j/\sqrt{b}_j$, using vertex-disjoint paths. Since the sub-CDAGs $G^{(j)}$ are vertex disjoint, so are the paths connecting vertices in $Y^{(j)}$ to vertices in $K^{(j)}$. 
    In the following we show that it is indeed possible to extended at least $\min\{2M,\sum_{j=1}^7 a_j / \sqrt{\sum_{j=1}^7 b_j}\}$ of these paths to vertices in $\globalInput{}$ while maintaining them vertex disjoint. 
    
	According to the construction of $G^{\mathcal{A}}$ as discussed in Section~\ref{sec:CDAG}, vertices in $\globalInput{}$ corresponding to the entries of input matrix $\mathbf{A}$ (resp., $\mathbf{B}$) are connected to vertices in \sloppy $\mathcal{K}^{(1)},\mathcal{K}^{(2)},\ldots,\mathcal{K}^{(7)}$ (and, hence, $K^{(1)},K^{(2)},\ldots,K^{(7)}$)  by means of $n^2$ encoding sub-CDAGs $Enc_A$ (resp., $Enc_B$). None of these $2n^2$ encoding sub-CDAGs share any input or output vertices. No two output vertices of the same encoder sub-CDAG belong to the same sub-CDAG $G^{(j)}$, for $j=1,2,\ldots,7$. This fact ensures that for a single sub-CDAG $G^{(j)}$, for $j=1,2,\ldots,7$, it is possible to connect all the vertices in $\mathcal{K}^{(j)}$ (and, hence, $K^{(j)}$) to a subset of the vertices in $\globalInput{}$ via vertex disjoint paths.
    
	For each of the $2n^2$ encoder sub-CDAGs, let us consider the vector $\mathbf{y}_l\in\{0,1\}^7$ such that $\mathbf{y}_l[j] = 1$ iff the corresponding $j$-th output vertex of the encoder, which is an input of $G^{(j)}$, is in $K^{(j)}$. Therefore $|\mathbf{y}_l|$ equals the number of output vertices of the $l$-th encoder sub-CDAG which are in $K$.
	 From Lemma~\ref{lem:conneconder}, for each encoder sub-CDAG there exists a subset $X_l\in\globalInput$ of the input vertices of the $l$-th encoder sub-CDAG for which it is possible to connect each vertex in $X_l$ to a distinct output vertex of the $l$-th encoder sub-CDAG using vertex disjoint paths, each constituted by a singular edge with $\min\{|\mathbf{y}_l|, 1 +\lceil\left(|\mathbf{y}_l|-1\right)/2\rceil\}\leq |X_l|\leq |\mathbf{y}_l|$. The number of vertex disjoint paths connecting vertices in $\globalInput{}$, to vertices in $\cup_{j=1}^7 K^{(j)}$ is therefore at least $\sum_{l=1}^{2n^2} \min\{|\mathbf{y}_l|, 1 +\lceil\left(|\mathbf{y}_l|-1\right)/2\rceil\}$, under the constraint that $\sum_{l=1}^{2n^2} \mathbf{y}_l[j]= a_j/\sqrt{b_j}$, for $j=1,2,\ldots,7$.

		 Let us assume w.l.o.g. that $a_1/\sqrt{b_1} \geq a_2/\sqrt{b_2}\geq \ldots\geq a_7/\sqrt{b_7}$. As previously stated, it is possible to connect all vertices in $K_1$ to vertices in $\globalInput{}$ through vertex disjoint paths. Consider now all possible dispositions of the vertices in $\cup_{j=2}^7 K^{(j)}$ over the outputs of the $2n^2$ encoder sub-CDAGs. Hence, if $a_1/\sqrt{b_1}\geq 2M$ we have that there are therefore at least $M$ vertex disjoint paths connecting vertices in $\globalInput{}$ to vertices in $K_1$, and, thus, to vertices in $Y$ as desired.
		 In the following we assume $a_1/\sqrt{b_1}< M$.
		 
		Recall that the output vertices of an encoder sub-CDAG belong each to a different sub-CDAG $G^{(j)}$. From Lemma~\ref{lem:conneconder} we have that for each encoder there exists a subset $X_l\subset{X}$ of the input vertices of the $l$-th encoder sub-CDAG, with
		$
			|X_l|\geq \min \Big\{|\mathbf{y}_l|, 1 + \left\lceil\left(|\mathbf{y}_l|-1\right)/2\right\rceil \Big\} \geq \mathbf{y}_l[1] + \left(\sum_{j=2}^7 \mathbf{y}_l[j]\right)/2
		$, 		for which is possible to connect all vertices in $X_l$ to $|X_l|$ \emph{distinct} output vertices of the $l$-th encoder sub-CDAG which are in $\cup_{j=1}^7 K^{(j)}$ using $|X_l|$ vertex disjoint paths.
		As all the $Enc$ sub-CDAGs are vertex-disjoint, we can add their contributions so that the number of vertex-disjoint paths connecting vertices in $\globalInput{}$ to vertices in  $\cup_{j=1}^7 K^{(j)}$ is at least 
		\begin{align*}
		    |K^{(1)}|+ \frac{1}{2}\sum\limits_{j=2}^7|K^{(j)}| &= \frac{a_1}{\sqrt{b_1}} \frac{1}{2}\left(\sum\limits_{j=2}^7 \frac{a_j}{\sqrt{b_j}}\right)\\
		    &\geq \frac{\sum_{j=1}^7 a_j}{\sqrt{\sum_{j=1}^7 b_j}},
		\end{align*}
		where the last passage follows by applying  Lemma~\ref{lem:lastminute}. There are therefore at least $a/\sqrt{b}$ vertex disjoint paths connecting vertices in $\globalInput{}$ to vertices in $\cup_{j=1}^7 K^{(j)}$, and, thus, to vertices in $Y$ as desired. This concludes the proof for case \textbf{(a)}.\\

		For case \textbf{(b)}, only one of the seven sub-problems $P^{(j)}$ generates all $\nu_1$ Type 1 MSPs. Without loss of generality, let $P^{(1)}$ denote such sub-problem and let $G^{\mathcal{A}_{P^{(1)}}}$ be the corresponding sub-CDAG. According to the construction of $G^{\mathcal{A}}$ as discussed in Section~\ref{sec:CDAG}, vertices in $\globalInput{}$ corresponding to the entries of input matrix $\mathbf{A}$ (resp., $\mathbf{B}$) are connected to the input vertices of $G^{\mathcal{A}_{P^{(1)}}}$, by means of $n^2$ encoding sub-CDAGs $Enc_A$ (resp., $Enc_B$). None of these $2n^2$ encoding sub-CDAGs share any input or output vertices. No two output vertices of the same encoder sub-CDAG belong to the same sub-CDAG $G^{(j)}$, for $j=1,2,\ldots,7$. This fact ensures that it is possible to connect all the input vertices of $G^{\mathcal{A}_{P^{(1)}}}$ to a subset of the vertices in $\globalInput{}$ via vertex disjoint paths. The proof for case \textbf{(b)} then follows by recursively applying the arguments in the Proof of Lemma~\ref{lem:domtype2} to $G^{\mathcal{A}_{P^{(1)}}}$.
\end{proof}

\subsection{Proof of Lemma~\ref{lem:domtype1}}\label{app:domtype1}
Given an algorithm $\mathcal{A}\in \hmm{}$ used to multiply input squared matrices $\mathbf{A}, \mathbf{B} \in \ri^{n \times n}$, let $G^{\mathcal{A}}$ denote the corresponding CDAG constructed according to the description in Section~\ref{sec:CDAG}. In the following we denote as $\globalInput{}$ the set of input vertices of $G^{\mathcal{A}}$. That is the set of vertices corresponding each to the entries in the input matrices $\mathbf{A}$ and $\mathbf{B}$. Further, for each Type 2 MSP $P_i$ we denote as $\subpInput{}_i$ the set of input vertices of the sub-CDAG $G^{\mathcal{A}}_{P_i}$ associated with $P_i$. That is the set of vertices corresponding each to the entries in the input matrices $\mathbf{A}_i$ and $\mathbf{B}_i$ of $P_i$. Also, we define $\subpInput{} = \cup_{i=1}^{\nu_2} \subpInput{}_i$.

In order to simplify the presentation of the proof of Lemma~\ref{lem:domtype1} we first introduce Lemma~\ref{lem:stra_part1} which is a heavily modified version of a result previously introduced in~\cite[Lemma 6]{bilardi2017complexity}.

\begin{lemma}\label{lem:stra_part1}
Consider an algorithm $\mathcal{A}\in \hmm{}$ with input matrices $\mathbf{A},\mathbf{B}\in \ri{}^{n\times n}$ which generates $\nu_2$ Type 2 MSPs.
Given the corresponding CDAG $G^{\mathcal{A}}$, let $\setproof{}$ be a set of \emph{internal} (i.e., not input) vertices of its $\nu_2$ sub-CDAGs corresponding each to one of the generated Type 2 MSPs. For any  $Z\subseteq \subpOutput{}$ with $|Z|\geq 2|\setproof{}|$ there exist $X\subseteq\globalInput{}$  with $|X|\geq 2\sqrt{M\left(|Z|-2|\setproof{}|\right)}$ such that each vertex in $X$ is connected to some vertex in $Z$ by a directed path with no vertex in $\setproof{}$.
\end{lemma}
\begin{proof}
The proof proceeds by induction on the number $\nu_2$ of Type 2 MSPs generated by $\mathcal{A}$.
In the base case $\nu_2=1$ and, by Definition~\ref{def:maxsubp} the entire problem is the only, improper, Type 2 MSP generated by $\mathcal{A}$ and, thus, the sets
$\subpInput{}$ and $\globalInput{}$ coincide. Consider now the following lemma:

\begin{lemma}[{\cite[Lemma 5]{bilardi2017complexity}}]\label{lem:base_connectivity}
Let $G^{n \times n}$ denote the CDAG corresponding to the execution of an \emph{unspecified} algorithm for the square matrix multiplication function with input matrices of size $n\times n$. Let $O'\subseteq O$ be a subset of its output vertices $O$. For any subset $\dom{}$ of the vertices of $G^{n\times n}$
with $|O'|\geq 2|\dom{}|$, there exists a set $I'\subseteq I$ of the input vertices $I$ of $G^{n\times n}$ with cardinality $|I'|\geq 2n\sqrt{|O'|-2|\dom{}|}$, such that all vertices in $I'$ are connected to some vertex in $O'$ by directed paths with no vertex in $\dom{}$.
\end{lemma}

As, by Definition~\ref{def:maxsubp}, $G^{\mathcal{A}}$ is a $G^{n\times
  n}$ CDAG, with $n\geq 2\sqrt{M}$, we can conclude that the statement of Lemma~\ref{lem:stra_part1} is verified in the base case as a consequence of Lemma~\ref{lem:base_connectivity}  and of the fact that, as previously mentioned, $\subpInput{}$ and $\globalInput{}$ coincide,
  
Lemma~\ref{lem:base_connectivity} was originally introduced in \cite[Lemma 5]{bilardi2017complexity} and is based on the analysis of the \emph{Grigoriev's flow} of the matrix multiplication function. For the sake of completeness we present the proof in Appendix~\ref{proof:basecon}.\\

Assuming now inductively that the statement holds for $\nu_2=k>1$,  we shall show
it also holds for $\nu_2 = k+1$.  As $\nu_2>1$, we have that the algorithm executes at least one recursive step. Let $P^{(j)}$ denote the 7 sub-problems generated at the first recursion step. We distinguish two cases:
\textbf{(a)} at least two of of the seven sub-problems $P^{(j)}$ generate each at least one Type 2 MSP; \textbf{(b)} only one of the seven sub-problems generates all $\nu_2$ Type 2 MSPs. We first address case \textbf{(a)}, as case \textbf{(b)} will follow from a simple extension.

For case \textbf{(a)}, let $G^{(j)}$, for $j=1,2,\ldots,7$ denote the seven sub-CDAGs of $G^{\mathcal{A}}$, each corresponding to one of the seven sub-problems generated in the first recursive step of $\mathcal{A}$ according to the chosen Strassen-like scheme as discussed in Section \ref{sec:CDAG}. By Definition~\ref{def:maxsubp} and as, by assumption, $\mathcal{A}$ generates Type 2 MSPs, we have that $n/2\geq 2\sqrt{M}$. Further, each of the seven sub-problems $P^{(j)}$ generates at most $\nu_2-1$ Type 2 MSPs.

Let $Z^{(j)}$, $\subpInput{}^{(j)}$ and $\setproof{}^{(j)}$ respectively denote the subsets of $Z$, $\subpInput{}$ and $\setproof{}$ in $G^{(j)}$, for $j=1,2,\ldots,7$. By  Lemma~\ref{lem:distinct}  the $G^{(j)}$'s have distinct input values and, hence, are pairwise vertex-disjoint sub-CDAGs of $G^{\mathcal{A}}$. Thus, $Z_1,Z_2,\ldots,Z_7$ partition $Z$, $\subpInput{}_1,\subpInput{}_2,\ldots,\subpInput{}_7$ partition $\subpInput{}$ and $\setproof{}_1,\setproof{}_2,\ldots,\setproof{}_7$ partition $\setproof{}$.  This implies $\sum_{j=1}^7 |Z^{(j)}| = |Z|$ and $\sum_{j=1}^7 |\setproof{}^{(j)}| = |\setproof{}|$. Let $\delta^{(j)} = \max\{0, |Z^{(j)}|-2|\setproof{}^{(j)}|\}$, we have $\delta = \sum_{j=1}^7 \delta^{(j)} \geq |Z| - 2|\setproof{}|$. 
	
	Applying the inductive hypothesis to each $G^{(j)}$, we have that there is a subset $Y^{(j)}\subseteq \subpInput{}^{(j)}$ with $|Y^{(j)}|\geq 4\sqrt{M\delta^{(j)}}$ such that vertices of $Y^{(j)}$ are connected  to vertices in $Z^{(j)}$ via paths with no vertex in $\setproof{}^{(j)}$.
    In the sequel the set $Y$ referred to in the statement will be identified as a suitable subset of $\cup_{i=j}^7 Y^{(j)}$ so that property (b) will be automatically satisfied. Towards property (a), we observe by the inductive hypothesis that vertices in $Y^{(j)}$ can be connected to a subset $K^{(j)}$ of the input vertices of $G^{(j)}$ with $|K^{(j)}|=|Y^{(j)}|$, using vertex-disjoint paths. Since the sub-CDAGs $G^{(j)}$ are vertex disjoint, so are the paths connecting vertices in $Y^{(j)}$ to vertices in $K^{(j)}$. It remains to show that at least $4\sqrt{M\left(|Z|-2|\setproof{}|\right)}$ of these paths can be extended to  $\globalInput{}$ while maintaining them vertex-disjoint. 
    
	According to the construction of $G^{\mathcal{A}}$ as discussed in Section~\ref{sec:CDAG}, vertices in $\globalInput{}$ corresponding to the entries of input matrix $\mathbf{A}$ (resp., $\mathbf{B}$) are connected to vertices in \sloppy $K_1,K_2,\ldots,K_7$ by means of $n^2$ encoding sub-CDAGs $Enc_A$ (resp., $Enc_B$). None of these $2n^2$ encoding sub-CDAGs share any input or output vertices. No two output vertices of the same encoder sub-CDAG belong to the same sub-CDAG $G^{(j)}$, for $j=1,2,\ldots,7$. This fact ensures that for a single sub-CDAG $G^{(j)}$, for $j=1,2,\ldots,7$, it is possible to connect all the vertices in $K_i$ to a subset of the vertices in $\globalInput{}$ via vertex disjoint paths.
    
	For each of the $2n^2$ encoder sub-CDAGs, let us consider the vector $\mathbf{y}_l\in\{0,1\}^7$ such that $\mathbf{y}_l[j] = 1$ iff the corresponding $j$-th output vertex of the encoder, which is an input of $G^{(j)}$, is in $K^{(j)}$. Therefore $|\mathbf{y}_l|$ equals the number of output vertices of the $l$-th encoder sub-CDAG which are in $K$.
	 From Lemma~\ref{lem:conneconder}, for each encoder sub-CDAG there exists a subset $X_l\in\globalInput$ of the input vertices of the $l$-th encoder sub-CDAG for which it is possible to connect each vertex in $X_l$ to a distinct output vertex of the $l$-th encoder sub-CDAG using vertex disjoint paths, each constituted by a singular edge with $\min\{|\mathbf{y}_l|, 1 +\lceil\left(|\mathbf{y}_l|-1\right)/2\rceil\}\leq |X_l|\leq |\mathbf{y}_l|$. The number of vertex disjoint paths connecting vertices in $\globalInput{}$, to vertices in $\cup_{j=1}^7 K^{(j)}$ is therefore at least $\sum_{l=1}^{2n^2} \min\{|\mathbf{y}_l|, 1 +\lceil\left(|\mathbf{y}_l|-1\right)/2\rceil\}$, under the constraint that $\sum_{l=1}^{2n^2} \mathbf{y}_l[j]= 4\sqrt{M\delta^{(j)}}$, for $j=1,2,\ldots,7$.
		 Let us assume w.l.o.g. that $\delta^{(1)} \geq \delta^{(20}\geq \ldots\geq \delta^{(7)}$. As previously stated, it is possible to connect all vertices in $K_1$ to vertices in $\globalInput{}$ through vertex disjoint paths. Consider now all possible dispositions of the vertices in $\cup_{j=2}^7 K^{(j)}$ over the outputs of the $2n^2$ encoder sub-CDAGs. 
		Recall that the output vertices of an encoder sub-CDAG belong each to a different sub-CDAG $G^{(j)}$. From Lemma~\ref{lem:conneconder} we have that for each encoder there exists a subset $X_l\subset{X}$ of the input vertices of the $l$-th encoder sub-CDAG, with
		$
			|X_l|\geq \min \Big\{|\mathbf{y}_l|, 1 + \left\lceil\left(|\mathbf{y}_l|-1\right)/2\right\rceil \Big\} \geq \mathbf{y}_l[1] + \left(\sum_{j=2}^7 \mathbf{y}_l[j]\right)/2
		$, 		for which is possible to connect all vertices in $X_l$ to $|X_l|$ \emph{distinct} output vertices of the $l$-th encoder sub-CDAG which are in $\cup_{j=1}^7 K^{(j)}$ using $|X_l|$ vertex disjoint paths.
		As all the $Enc$ sub-CDAGs are vertex disjoint, we can add their contributions so that the number of vertex disjoint paths connecting vertices in $\globalInput{}$ to vertices in  $\cup_{j=1}^7 K^{(j)}$ is at least $|K_1|+ \frac{1}{2}\sum\limits_{j=2}^7|K^{(j)}| = 4\sqrt{M}\left(\sqrt{\delta^{(1)}} + \frac{1}{2}\sum\limits_{j=2}^7\sqrt{\delta^{(j)}} \right)$.
	Squaring this quantity leads to:
	\begin{equation*}
		\left(4\sqrt{M}\left(\sqrt{\delta^{(1)}} +\frac{1}{2}\sum\limits_{j=2}^7\sqrt{\delta^{(j)}} \right)\right)^2 = 16M\left(\delta^{(1)} +\sqrt{\delta^{(1)}}\sum\limits_{j=2}^7\sqrt{\delta^{(j)}} + \left(\frac{1}{2}\sum\limits_{j=2}^7\sqrt{\delta^{(j)}}\right)^2\right).
	\end{equation*}
	As, by assumption, $\delta^{(1)} \geq\ldots \delta^{(7)}$, we have: $\sqrt{\delta^{(1)}}\sqrt{\delta^{(j)}}\geq \delta^{(j)}$ for $j=2,\ldots,7$. Thus:
	\begin{equation*}
		\left(4\sqrt{M}\left(\sqrt{\delta^{(1)}} +\frac{1}{2}\sum\limits_{j=2}^7\sqrt{\delta^{(j)}} \right)\right)^2 
		\geq 16 M \sum\limits_{j=1}^7 \delta^{(j)} \geq \left(4\sqrt{M\left(|Z|-2|\setproof{}|\right)}\right)^2.
	\end{equation*}
	There are therefore at least $4\sqrt{M\left(|Z|-2|\setproof{}|\right)}$ vertex disjoint paths connecting vertices in $\globalInput{}$ to vertices in $\cup_{j=2}^7 K^{(j)}$ as desired. This concludes the proof for case \textbf{(a)}.\\
	
	For case \textbf{(b)}, only one of the seven sub-problems $P^{(j)}$ generates all $\nu_1$ Type 1 MSPs. Without loss of generality, let $P^{(1)}$ denote such sub-problem and let $G^{\mathcal{A}_{P^{(1)}}}$ be the corresponding sub-CDAG. According to the construction of $G^{\mathcal{A}}$ as discussed in Section~\ref{sec:CDAG}, vertices in $\globalInput{}$ corresponding to the entries of input matrix $\mathbf{A}$ (resp., $\mathbf{B}$) are connected to the input vertices of $G^{\mathcal{A}_{P^{(1)}}}$, by means of $n^2$ encoding sub-CDAGs $Enc_A$ (resp., $Enc_B$). None of these $2n^2$ encoding sub-CDAGs share any input or output vertices. No two output vertices of the same encoder sub-CDAG belong to the same sub-CDAG $G^{(j)}$, for $j=1,2,\ldots,7$. This fact ensures that it is possible to connect all the input vertices of $G^{\mathcal{A}_{P^{(1)}}}$ to a subset of the vertices in $\globalInput{}$ via vertex disjoint paths. The proof for case \textbf{(b)} then follows by recursively applying the arguments in the Proof of Lemma~\ref{lem:stra_part1} to $G^{\mathcal{A}_{P^{(1)}}}$.
\end{proof}

Lemma~\ref{lem:stra_part1}, provides the base for the proof of Lemma~\ref{lem:domtype1}, which is itself a reworked version a result from~\cite{bilardi2017complexity} (Lemma 7) modified in order to account for the hybrid nature of algorithms being considered in this work.
\begin{proof}[Proof of Lemma~\ref{lem:domtype1}]
	Suppose for contradiction that $\dom$ is a dominator set for $Z$ in $G^{\mathcal{A}}$ such that $|\dom|\leq 2M-1$. 
	Let $\dom'\subseteq \dom$ be the subset of the vertices of $\dom $ composed by  vertices which are \emph{not} internal to the sub-CDAGs corresponding to the Type 2 MSPs generated, by assumption, by $\mathcal{A}$. 	
	From Lemma~\ref{lem:stra_part1}, with $Q = \dom \setminus \dom'$, there exist  $X \subseteq \globalInput{}$ and $Y\subseteq \subpInput$ with  $|X|=|Y|\geq 4\sqrt{M\left(|Z|-2\left(|\dom|-|\dom'|\right)\right)}$ such that vertices in $X$ are connected to vertices in $Y$ by vertex-disjoint paths. Hence, each vertex in $\dom'$  can be on at most one of these paths. Thus, there exists $X'\subseteq X$ and $Y'\subseteq Y$ with $|X'|=|Y'|\geq \phi=  4\sqrt{M\left(|Z|-2\left(|\dom| - |\dom'|\right)\right)} - |\dom'|$  paths from $X'$ to $Y'$ with no vertex in $\dom'$. From Lemma~\ref{lem:stra_part1}, we also have that all vertices in $Y$, and, hence, in $Y'$, are connected to some vertex in $Z$ by a path with no vertex in $\dom\setminus \dom'$. Thus, there are at least $\phi$ paths connecting vertices in $X'\subseteq \mathcal{X}$ to vertices in $Z$ with no vertex in $\dom$.  We shall now show that the contradiction assumption  $|\dom|\leq 2M-1$ implies $\phi>0$:
\begin{align*}
		\left( 4\sqrt{M\left(|Z|-2\left(|\dom| - |\dom'|\right)\right)}\right)^2 &= 16 M \left(|Z|-2\left(|\dom| - |\dom'|\right)\right), \\
		&= 16 M \left(|Z|-2|\dom|\right) + 32 M  |\dom'|.
\end{align*}
By $|\dom|\leq 2M-1$,  we have $|Z|-2|D|>4M-2(M-1)>0$. Furthermore, from $\dom'\subseteq \dom$, we have $32M > 2M-1>|\dom|\geq |\dom'|$. Therefore:
\begin{equation}\label{eq:util}
		\left(\phi + |\dom'|\right)^2=\left( 4\sqrt{M\left(|Z|-2\left(|\dom| - |\dom'|\right)\right)}\right)^2 
		> |\dom'|^2.
 \end{equation}
Again, $|\dom|\leq 2M-1$ implies $M\left(|Z|-2\left(|\dom| - |\dom'|\right)\right)> 0$. Hence, we can take the square root on both sides of (\ref{eq:util}) and conclude that $\phi > 0$. Therefore, for $|\dom|\leq 2M-1$ there are at least $\phi>0$ paths connecting a global input vertex to a vertex in $Z$ with no vertex in $\dom$, contradicting the assumption that $\dom$ is a dominator of $Z$. 
\end{proof}

\subsection{Proof of Lemma~\ref{lem:domtype3}}\label{app:domtype3}
\begin{proof}[Proof of Lemma~\ref{lem:domtype3}]
Without loss of generality, let us assume $ |\subpInput{}_i^{(\mathbf{A})}|\geq|\subpInput{}_i^{(\mathbf{B})}|$. The proof for the case $|\subpInput{}_i^{(\mathbf{A})}|<|\subpInput{}_i^{(\mathbf{B})}|$ follows an analogous argument.
Let $\dom{}$ be a dominator set for the set of vertices corresponding to $T_i'$ with respect to
$\subpInput{}_i^{(\mathbf{A})}$. 

Consider a possible assignment to the values of $B_i$ such that all the values corresponding to vertices in $\subpInput{}_i^{(\mathbf{B})}$ are assigned value 1. 
Under such assignment, for every variable $a\in A_i$ corresponding to a vertex in $\subpInput{}_i^{(\mathbf{A})}$ at least one of the elementary products in $T_i'$ assumes value $a$. The lemma follows combining statements (i) and (ii):\\ (i) There exists an
assignment of the input variables of $P_i$ corresponding to vertices in $\subpInput{}_i\setminus \subpInput{}_i^{(\mathbf{A})}$
such that the output variables in $T_i'$ assume at least
$|\ri|^{|\subpInput{}_i^{(\mathbf{A})}|}$ distinct values under all possible assignments of the variables corresponding to vertices in $\subpInput{}_i^{(\mathbf{A})}$.\\ (ii) Since all paths form $\subpInput{}_i^{(\mathbf{A})}$ to the vertices corresponding to the variables in $T_i'$ intercept $\dom$, the values of the elementary products in $T_i'$ are determined by the inputs in
$\subpInput{}_i\setminus \subpInput{}_i^{(\mathbf{A})}$, which are fixed, and by the values of the
vertices in $\dom{}$; hence the elementary products in $T_i'$ can take at most
$|\ri|^{|\dom|}$ distinct values. 
\end{proof}

\subsection{Proof of Lemma~\ref{lem:base_connectivity}}\label{proof:basecon}
Before presenting the proof of Lemma~\ref{lem:base_connectivity}, we present the concept of ``\emph{flow of a function}'' which was originally introduced by Grigoriev~\cite{grigor1976application} and then revised by Savage~\cite{savage97models}. In this work, we use the following version:

\begin{definition}[Grigoriev's flow]\label{def:flow}
A function $f:\mathcal{\ri}^p\rightarrow\mathcal{\ri}^q$ has a $w\left(u, v\right)$ Grigoriev's flow if for all subsets $X_1$ (resp., $Y_1$), of its $p$ input (resp., $q$ output) variables, with $|X_1| \geq u$ and $|Y_1| \geq  v$, there is a sub-function $h$ of $f$ obtained by making some assignment to variables of $f$ not in $X_1$  and discarding output variables not in $Y_1$, such that $h$ has at least $|\ri|^{w(u,v)}$ points in the image of its domain. 	
\end{definition}

Grigoriev's flow is an inherent property of a function, agnostic to the algorithm used to compute the function. It provides a lower bound to the amount of information that suitable sub-sets of outputs encode about suitable sub-sets of inputs. Since any information about inputs that is encoded by outputs must be also be encoded by any dominator of those outputs, we have the following lower bound on the size of dominator sets.

\begin{lemma}[{\cite[Lemma 4]{bilardi2017complexity}}]\label{lem:flowdom}
Let $G=(V,E)$ be a CDAG computing a function $f:\mathcal{\ri}^p\rightarrow\mathcal{\ri}^q$, defined on the ring $\mathcal{R}$, with $p$ input and $q$ output variables. Let $I$ (resp., $O$) denote the set of input (resp.,
output) vertices of $G$. Further let $I_X$ (resp., $O_Y$) denote the subset of $I$ (resp., $O$) which correspond to a subset $X$ (resp., $Y$) of the $p$ input (resp., $q$ output) variables of $f$. 
For any $O'\subseteq O_Y$  and any $I'\subseteq I_X$, any dominator of $O'$ with respect to $I'$ satisfies $|\dom| \geq w_{f_{X,Y}}(|I'|,|O'|)$.
\end{lemma}

A lower bound on the Grigoriev's flow for the square matrix multiplication function $f_{n \times n}: \ri^{2n^2} \rightarrow \ri^{n^2}$ over the ring $\ri$ was presented  in~\cite[Theorem 10.5.1]{savage97models}. 
 
\begin{lemma}[Grigoriev's flow of $f_{n \times n}: \ri^{2n^2} \rightarrow \ri^{n^2}$ ~\cite{savage97models}]\label{lem:info_flo_mat_mul}
	$f_{n \times n}: \ri^{2n^2} \rightarrow \ri^{n^2}$
 has a $w_{n\times n}\left(u, v\right)$ Grigoriev's flow, where:
	\begin{equation}
		w_{n \times n}\left(u, v\right) \geq \frac{1}{2}\left(v-\frac{\left(2n^2 -u\right)^2}{4n^2}\right),\mathrm{ for}\ 0\leq u \leq 2n^2,\ 0\leq v\leq n^2.
	\end{equation}
\end{lemma}

We can now state the proof of Lemma~\ref{lem:base_connectivity}. Recall that we denote as $G^{n \times n}$  the CDAG corresponding to the execution of a \emph{unspecified} algorithm for the square matrix multiplication function with input matrices of size $n\times n$.

\begin{proof}[Proof of Lemma~\ref{lem:base_connectivity}] The statement follows by applying the results in Lemma~\ref{lem:flowdom} and Lemma~\ref{lem:info_flo_mat_mul} to the CDAG $G^{n \times n}$. Let $I''\subseteq I$ denote the set of all input vertices of $G^{n\times n},$ such that all paths connecting these vertices  to the output vertices in $O'$ include at least a vertex in $\dom{}$ (i.e., $I''$ is the largest subset of $I$ with respect to whom $\dom{}$ is a dominator set for $O'$). 
	From Lemmas~\ref{lem:info_flo_mat_mul} and~\ref{lem:flowdom} the following must hold:
	\begin{equation}\label{eq:newlem1}
		|\dom|\geq w_{n\times n} \geq 
		\frac{1}{2}\left(|O'|-\frac{\left(2n^2 - |I''|\right)^2}{4n^2}\right).
	\end{equation}	
	Let $I' =I \setminus I''$. By the definition of $I''$, the vertices in $I'$ are exactly those that are connected to vertices in $O'$ by directed paths with no vertex in $\dom{}$. Since $|I|= 2n^2$, from (\ref{eq:newlem1}) we have $|I'|^2  \geq 4n^2\left(|O'|-2|\dom{}|\right)$.
\end{proof}

\clearpage
\section{Strassen's fast multiplication algorithm}\label{app:strassenalg}
\begin{algorithm}
\caption{Strassen's Matrix Multiplication}\label{alg:strass}
\begin{algorithmic}[1]
	\Statex \textbf{Input:} matrices $A,B$
	\Statex \textbf{Output:} matrix $C$
	\Procedure{StrassenMM}{A,B}
	\If{$n=1$}
	\State $C=A\cdot B$
	\Else
	\State Decompose $A$ and $B$ into four equally sized block matrices as follows:
	\Statex \begin{equation*}
\begin{array}{cc}
	A = \left[ \begin{array}{cc}
	A_{1,1} &A_{1,2} \\
	A_{2,1} &A_{2,2}
	\end{array}
	\right], & B = \left[ \begin{array}{cc}
	B_{1,1} &B_{1,2} \\
	B_{2,1} &B_{2,2}
	\end{array}
	\right]
\end{array}
\end{equation*}
\State $M_1 = \textsc{StrassenMM}\left(A_{1,1} + A_{2,2},B_{1,1} + B_{2,2}\right)$
\State $M_2 = \textsc{StrassenMM}\left(A_{2,1} + A_{2,2}, B_{1,1}\right)$
\State $M_3 = \textsc{StrassenMM}\left(A_{1,1},B_{1,2} - B_{2,2}\right)$
\State $M_4 = \textsc{StrassenMM}\left(A_{2,2},B_{2,1} - B_{1,1}\right)$
\State $M_5 = \textsc{StrassenMM}\left(A_{1,1} + A_{1,2}, B_{2,2}\right)$
\State $M_6 = \textsc{StrassenMM}\left(A_{2,1} - A_{1,1}, B_{1,1} + B_{1,2}\right)$
\State $M_7 = \textsc{StrassenMM}\left(A_{1,2} - A_{2,2}, B_{2,1} + B_{2,2}\right)$
\State $C_{1,1} = M_1 +M_4-M_5+M_7$ 
\State $C_{1,2} = M_3 + M_5$
\State $C_{2,1} = M_2 + M_4$
\State $C_{2,2} = M_1 -M_2 +M_3 +M_6$
\EndIf
\Return $C$
\EndProcedure
\end{algorithmic}
\end{algorithm}

The original version of  Strassen's fast matrix multiplication~\cite{strassen1969gaussian} is reported in Algorithm~\ref{alg:strass}. We refer the reader to~\cite{winograd1971multiplication} for
Winograd's variant, which reduces the number of additions. 

\begin{figure}[t!]
\begin{subfigure}{.31\textwidth}
  \centering
        \resizebox{\linewidth}{!}{
            \begin{tikzpicture}[scale = 0.5,
            > = stealth, 
            shorten > = 1pt, 
            auto,
            semithick 
        ]

        \tikzstyle{every state}=[
            draw = black,
            thick,
            minimum size = 1mm
        ]

        \node[state] at (0,0) (11) [label=below:$A_{1,1}$] {};
        \node[state] at (2,0) (12)  [label=below:$A_{1,2}$] {};
        \node[state] at (4,0) (21)  [label=below:$A_{2,1}$] {};
        \node[state] at (6,0) (22)  [label=below:$A_{2,2}$] {};
        \node[state] at (-3,4) (7)  [label=above:$7$] {};
        \node[state] at (-1,4) (5)  [label=above:$5$] {};
        \node[state] at (1,4) (4)  [label=above:$4$] {};
        \node[state] at (3,4) (1)  [label=above:$1$] {};
        \node[state] at (5,4) (3)  [label=above:$3$] {};
        \node[state] at (7,4) (2)  [label=above:$2$] {};
        \node[state] at (9,4) (6)  [label=above:$6$] {};

        \path[->] (11) edge  (5);
        \path[->] (11) edge  (1);
        \path[->] (11) edge  (3);
        \path[->] (11) edge  (6);
        
        \path[->] (12) edge  (7);
        \path[->] (12) edge  (5);
        
        \path[->] (21) edge  (2);
        \path[->] (21) edge  (6);
        
        \path[->] (22) edge  (2);
        \path[->] (22) edge  (1);
        \path[->] (22) edge  (4);
        \path[->] (22) edge  (7);
    \end{tikzpicture}
        }
  \caption{$Enc_A$}
  \label{fig:enca}
\end{subfigure}
\hspace{4mm}
\begin{subfigure}{.31\textwidth}
  \centering
  \resizebox{\linewidth}{!}{
  \begin{tikzpicture}[scale = 0.5,
            > = stealth, 
            shorten > = 1pt, 
            auto,
            node distance = 3cm, 
            semithick 
        ]

        \tikzstyle{every state}=[
            draw = black,
            thick,
            fill = white,
            minimum size = 4mm
        ]

        \node[state] at (0,0) (11) [label=below:$B_{1,1}$] {};
        \node[state] at (2,0) (12)  [label=below:$B_{1,2}$] {};
        \node[state] at (4,0) (21)  [label=below:$B_{2,1}$] {};
        \node[state] at (6,0) (22)  [label=below:$B_{2,2}$] {};
        \node[state] at (-3,4) (7)  [label=above:$7$] {};
        \node[state] at (-1,4) (5)  [label=above:$5$] {};
        \node[state] at (1,4) (4)  [label=above:$4$] {};
        \node[state] at (3,4) (1)  [label=above:$1$] {};
        \node[state] at (5,4) (3)  [label=above:$3$] {};
        \node[state] at (7,4) (2)  [label=above:$2$] {};
        \node[state] at (9,4) (6)  [label=above:$6$] {};

        \path[->] (11) edge  (4);
        \path[->] (11) edge  (1);
        \path[->] (11) edge  (2);
        \path[->] (11) edge  (6);
        
        \path[->] (12) edge  (3);
        \path[->] (12) edge  (6);
        
        \path[->] (21) edge  (7);
        \path[->] (21) edge  (4);
        
        \path[->] (22) edge  (7);
        \path[->] (22) edge  (1);
        \path[->] (22) edge  (5);
        \path[->] (22) edge  (3);
    \end{tikzpicture}
    }
  \caption{$Enc_B$}
  \label{fig:encb}
\end{subfigure}
\begin{subfigure}{.31\textwidth}
  \centering
     \resizebox{\linewidth}{!}{
  \begin{tikzpicture}[scale = 0.5,
            > = stealth, 
            shorten > = 1pt, 
            auto,
            node distance = 3cm, 
            semithick 
        ]

        \tikzstyle{every state}=[
            draw = black,
            thick,
            fill = white,
            minimum size = 4mm
        ]

        \node[state] at (0,4) (C11) [label=above:$C_{1,1}$] {};
        \node[state] at (2,4) (C12)  [label=above:$C_{1,2}$] {};
        \node[state] at (4,4) (C21)  [label=above:$C_{2,1}$] {};
        \node[state] at (6,4) (C22)  [label=above:$C_{2,2}$] {};
        \node[state] at (-3,0) (7)  [label=below:$M_7$] {};
        \node[state] at (-1,0) (5)  [label=below:$M_5$] {};
        \node[state] at (1,0) (4)  [label=below:$M_4$] {};
        \node[state] at (3,0) (1)  [label=below:$M_1$] {};
        \node[state] at (5,0) (3)  [label=below:$M_3$] {};
        \node[state] at (7,0) (2)  [label=below:$M_2$] {};
        \node[state] at (9,0) (6)  [label=below:$M_6$] {};

        \path[<-] (C11) edge  (4);
        \path[<-] (C11) edge  (1);
        \path[<-] (C11) edge  (7);
        \path[<-] (C11) edge  (5);
        
        \path[<-] (C12) edge  (3);
        \path[<-] (C12) edge  (5);
        
        \path[<-] (C21) edge  (2);
        \path[<-] (C21) edge  (4);
        
        \path[<-] (C22) edge  (6);
        \path[<-] (C22) edge  (1);
        \path[<-] (C22) edge  (2);
        \path[<-] (C22) edge  (3);
    \end{tikzpicture}
    }
  \caption{$Dec$}
  \label{fig:dec}
\end{subfigure}
\caption[Encoder and Decoder sub-CDAGs corresponding to Strassen's original algorithm]{Encoder and Decoder sub-CDAGs corresponding to Strassen's original algorithm.}
\label{fig:fig}
\end{figure}
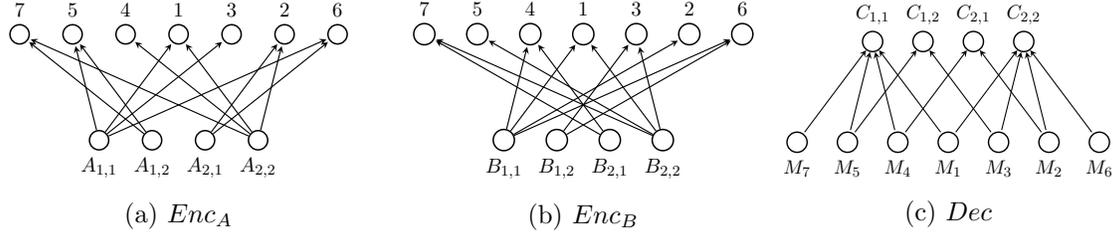